\newcount\Comments  
\documentclass[letterpaper]{article}
\usepackage{uai2018}
\usepackage[margin=1in]{geometry}
\usepackage{times}
\usepackage{latexsym}

\usepackage{amsmath,amsfonts,amsthm,amssymb}
\usepackage{graphicx}
\usepackage{subfigure}
\usepackage{algorithm}
\usepackage{color}
\usepackage{epstopdf}
\usepackage{url}
\usepackage{verbatim}
\usepackage{enumitem}
\usepackage{algorithmic}
\usepackage{natbib}
\usepackage[usenames,dvipsnames,svgnames,table]{xcolor}

\usepackage[colorlinks=true,
            linkcolor=NavyBlue,
            urlcolor=blue,
            citecolor=NavyBlue]{hyperref}
\usepackage{thm-restate}

\newtheorem{lemma}{Lemma}[section]

\newtheorem{definition}{Definition}[section]

\DeclareMathOperator*{\E}{\mathbb{E}}
\newcommand{\R}{\mathbb{R}}
\newcommand{\V}[1]{\mathbf{#1}}
\newcommand{\Sol}{\mathbf{x}}
\newcommand{\transpose}{{\textrm{T}}}
\newcommand{\B}{\widehat{b}}

\newcommand{\ignore}[1]{}
\newcommand{\kibitz}[2]{\ifnum\Comments=1\textcolor{#1}{#2}\fi}

\newcommand{\yc}[1]  {\kibitz{red}      {\bf\noindent [Yiling: #1]} }
\newcommand{\sherry}[1]  {\kibitz{green}      {\bf\noindent [Sherry: #1]} }

\title{Active Information Acquisition for Linear Optimization}

\author{ {\bf Shuran Zheng}\\%\thanks{Footnote for author to give an alternate address.}} \\
School of Engineering \\
and Applied Sciences \\
Harvard University\\
Cambridge, MA\\
\And
{\bf Bo Waggoner}\\ %\thanks{Footnote for author to give an alternate address.}}  \\
The Warren Center for \\
Network and Data Sciences\\
 Univerisity of Pennsylvania\\
 Philadelphia, PA\\
\And
{\bf Yang Liu}\\%\thanks{Footnote for author to give an alternate address.}}   \\
School of Engineering \\
and Applied Sciences \\
Harvard University\\
Cambridge, MA\\
\And
{\bf Yiling Chen} \\%\thanks{Footnote for author to give an alternate address.}}   \\
School of Engineering \\
and Applied Sciences \\
Harvard University\\
Cambridge, MA\\
}

\begin{document}
\maketitle

\begin{abstract}
  We consider partially-specified optimization problems where the goal is to actively, but efficiently, acquire missing information about the problem in order to solve it.
  An algorithm designer wishes to solve a linear program (LP), $\max \V{c}^T \V{x}$ s.t. $\V{A}\V{x} \leq \V{b}, \V{x} \ge \V{0}$, but does not initially know some of the parameters.
  The algorithm can iteratively choose an unknown parameter and gather information in the form of a noisy sample centered at the parameter's (unknown) value.
  The goal is to find an approximately feasible and optimal solution to the underlying LP with high probability while drawing a small number of samples.
  We focus on two cases.
  (1) When the parameters $\V{b}$ of the constraints are initially unknown, we propose an efficient algorithm combining techniques from the ellipsoid method for LP and confidence-bound approaches from bandit algorithms.
  The algorithm adaptively gathers information about constraints only as needed in order to make progress.
  We give sample complexity bounds for the algorithm and demonstrate its improvement over a naive approach via simulation.
  (2) When the parameters $\V{c}$ of the objective are initially unknown, we take an information-theoretic approach and give roughly matching upper and lower sample complexity bounds, with an (inefficient) successive-elimination algorithm.
  
\end{abstract}

\section{INTRODUCTION}
Many real-world settings are modeled as optimization problems. For example, a delivery company plans driver routes to minimize the driver's total travel time; an airline assigns vehicles to different origin-destination pairs to maximize profit. However, in practice, some parameters of the optimization problems may be initially unknown. The delivery company may not know the average congestion or travel time of various links of the network, but has ways to poll Waze\footnote{https://www.waze.com} drivers to get samples of travel times on links of the network. The delivery company may not know the demand and the revenues for each origin-destination pair, but can get estimates of them by selling tickets on chosen origin-destination pairs.  
%but some parameters of the problem are initially unknown. For example, consider a delivery company that wished to plan driver routes to minimizes the driver's total travel time. The 
%For example, consider a delivery company that wishes to plan driver routes, but does not initially know average congestion or travel time of various links of the network or an airline that wants to assign vehicles, but does not know the demand and the revenues for every origin-destination pair.% a reinforcement learning agent that attempts to find the optimal policy without knowing the reward function at the beginning. 
\sherry{added two examples here}\yc{I take out the reinforcement learning example and changed the phrasing slightly.}

To capture such settings, this paper proposes a model of optimization wherein the algorithm can iteratively choose a parameter and draw a ``sample'' that gives information about that parameter; specifically, the sample is an independent draw from a subgaussian random variable centered at the true value of the parameter.
This models, for instance, observing the congestion on a particular segment of road on a particular day. 
Drawing each sample is presumed to be costly, so the goal of the algorithm is to draw the fewest samples necessary in order to find a solution that is approximately feasible and approximately optimal.

Thus, the challenge falls under an umbrella we term \emph{active information acquisition for optimization (AIAO)}.
%In such settings, the algorithm does not exactly face an explore-exploit tradeoff as in, for instance, multi-armed bandits problems: The decisions made during information-gathering do not affect the algorithm's ``payoff''; only the final solution does.
%The settings also differ significantly from information acquisition problems in machine learning, such as active learning: Our solution's quality is measured via a structured objective function rather than as an average or expected loss.
%Closely related are \emph{combinatorial pure exploration} problems\bo{cite}, but these works generally do not model the structure of the optimization problem itself (or consider a trivial optimization problem).
The key feature of the AIAO setting is the structure of the optimization problem itself, i.e. the objective and constraints.
The challenge is to understand how the difficulty of information acquisition relates to this underlying structure.
For example, are there information-theoretic quantities relating the structure to the sample complexity?
Meanwhile, the opportunity of AIAO is to exploit algorithms for the underlying optimization problem.
For example, can one interface with the algorithm to reduce sample complexity by only acquiring the information needed, as it is needed?

These are the questions investigated in this paper, which focuses on active information acquisition for linear optimization problems.
%The goal of this paper is to investigate how to leverage the structure of an optimizatoin
%Instead, the primary new feature of the setting, presenting both a challenge and an opportunity, is the structure of the optimization problem itself.
%Successful approaches require an optimization algorithm as a prerequisite.
%They should parsimoniously acquire the information required for the algorithm to solve the problem while ignoring the rest.
%This is the general template utilized in this paper, which focuses on active information acquisition for linear programming problems.
Specifically, we consider linear programs in the form
\begin{equation}
  \max_{\V{x}} \V{c}^T \V{x}, \textrm{ s.t. } \V{A}\V{x} \le \V{b},  \V{x}\geq \V{0}\label{eqn:lp-definition}
\end{equation}
with $\V{A} \in \mathbb R^{m \times n}, \V{c} \in \mathbb R^{n}$, and $\V{b} \in \mathbb R^{m}$. 
We will consider either the case that the $\V{b}$ in the constraints is unknown or else the case that the $\V{c}$ in the objective is unknown, with all other parameters initially known to the algorithm.
The algorithm can iteratively choose an unknown parameter, e.g. $b_i$, and draw a ``sample'' from it, e.g. observing $b_i + \eta$ for an independent, zero-mean, subgaussian $\eta$.
The algorithm must eventually output a solution $\V{x}$ such that, with probability $1-\delta$, $\V{A}\V{x} \leq \V{b} + \varepsilon_1\V{1}$ and $\V{c}^T \V{x} \geq \V{c}^T \V{x^*} - \varepsilon_2$, where $\V{x^*}$ is the optimal solution.
The goal is for the algorithm to achieve this while using as few total samples as possible.

There is a natural ``naive" or ``static'' approach: draw samples for all unknown parameters until they are known to high accuracy with high probability, then solve the ``empirical'' linear program.
However, we can hope to improve by leveraging known algorithms and properties of linear programs.
For example, in the case that $\V{b}$ is unknown, if a linear program has an optimal solution, it has an optimal solution that is an extreme point (a corner point of the feasible region); and at this extremal optimal solution, several constraints are binding. These suggest that it is more important to focus on binding constraints and to gather information on the differing objective values of extreme points. Algorithms developed in this paper leverage these properties of linear programs to decide how much information to acquire for each unknown parameter.  

%The setting considered in this paper, active information acquisition for optimization, is related at a high level to a large number of lines of work (in theoretical computer science \citep{balkanski2016power}, machine learning \citep{balcan2006agnostic,balcan2007margin,castro2008minimax}, and artificial intelligence \citep{braziunas2006computational,blum2004preference}) that deal with optimization and uncertainty.
%However, most of them solve a different conceptual problem, failing to model either the optimization aspect or the active aspect. At a technical level, some works \citep{chen2014combinatorial, gabillon2016improved, chen2016combinatorial, chen2016pure, chen17nearly} on the \emph{combinatorial pure exploration} (CPE) problem, are relatively close to our model. But these works focus on combinatorial optimization problems, which is very different from the linear optimization problems investigated in this work. \sherry{Added more detailed comparison with the CPE problem here.}

\yc{I took out the paragraph on the related work here since we now have moved the related work to the intro.}

\subsection{APPROACHES AND RESULTS} \sherry{shortened this part.}

%\yc{I want to change the section below to talk about the two settings and the two approaches that we take (extreme points vs. interior-point method).}

\paragraph{Two settings and our approaches.}
The paper investigates two settings: unknown $\V{b}$ but known $\V{c}$, and unknown $\V{c}$ but known $\V{b}$. 
We always suppose $\V{A}$ is known and assume that the linear program has an optimal solution.

It might initially appear that these cases are ``equivalent'' via duality theory, but we argue that the two cases are quite different when we do not know the parameters of LP exactly.
Given a \emph{primal} linear program of the form (\ref{eqn:lp-definition}), the \emph{dual} program is given by $\min_{\V{y}} \V{b}^T \V{y}$ s.t. $\V{A}^T \V{y} \geq \V{c}, \V{y} \geq \V{0}$, which is easily transformed into the maximization format of (\ref{eqn:lp-definition}).
In particular, the parameters $\V{c}$ in the objective function of a primal LP becomes the parameters in the constraints of the dual LP.
By duality theory, the (exact) optimal solutions to the primal and dual are connected by \emph{complementary slackness} conditions.
However, this approach breaks down in the approximate setting for two reasons.
First, approximately optimal solutions do not satisfy complementary slackness; and second, even knowing which constraints bind does not suffice to determine the optimal solution $\V{x^*}$ when some of the constraint or objective parameters are unknown.\footnote{Nor does knowing which constraints bind even necessarily help, as approximately satisfying them may still lead to large violations of other constraints.
    Thus, while we do not rule out some future approach that connects approximate solutions of the primal and dual, the evidence suggests to us that the two settings are quite different and we approach each differently in this paper.}
We hence take two different approaches toward our two settings. 

%%%% Yiling's paragraph
% We always suppose $\V{A}$ is known (see the discussion in Section \ref{sec:discussion}  \bo{reference} on how our algorithms may extend to unknown-$\V{A}$) and assume that the linear program has an optimal solution.  One may be tempted to ask why these two cases are not equivalent via duality theory. In standard linear programming with known parameters, the parameters $\V{c}$ in the objective function of a primal LP becomes the parameters $\V{b}$ in the constraints of the dual LP and duality theory has established that the primal and dual LPs have the same optimal value and their optimal solutions are connected by complementary slackness conditions. Hence, in standard linear programming, we can find an optimal solution to an LP by solving either the primal problem or the dual problem. But this equivalence is broken in the AIAO setting because of the uncertainty on prameters. While duality theory connects the \emph{exact} optimal solutions of the primal and dual problems, it doesn't connect approximately feasible and optimal solutions of the two problems, which are what we seek to achieve in AIAO.    

\paragraph{Unknown-$\V{b}$ case.}
In the unknown $\V{b}$ setting, the uncertainty is over the constraints. Our algorithm combines two main ideas: the ellipsoid method for solving linear programs, and multi-armed bandit techniques for gathering data. The ellipsoid algorithm only requires the information of one violated constraint at each iteration, if there exists a violated one. We then leverage the multi-armed bandits method to find the most violated constraint (if there exists one) using as few samples as possible.

We theoretically bound the number of samples drawn by our algorithm as a function of the parameters
of the problem. In simulations on generated linear programs, UCB-Ellipsoid far outperforms the naive approach of sampling all parameters to high accuracy, and approaches the performance of an oracle that knows the binding constraints in advance and needs only to sample these. In other words, the algorithm spends very few resources on unimportant parameters.

\paragraph{Unknown-$\V{c}$ case.}
In the unknown-objective setting, we know the feasible region exactly. Our algorithm focuses only on the set of extreme points of the feasible region. For each of the extreme point, there are a set of possible values for $\V{c}$ such that if $\V{c}$ takes any value in the set, this extreme point is the optimal solution to the LP.
The algorithm hence draws just enough samples to determine with high probability which is actually the case for the true $\V{c}$.

We define an information-theoretic measure, $Low(\mathcal{I})$ for an instance $\mathcal{I}$.
We show that this quantity essentially characterizes the sample complexity of a problem instance and we give an algorithm, not necessarily efficient, for achieving it up to low-order factors.

\section{RELATED WORK} \label{app:relate}

The setting considered in this paper, active information acquisition for optimization, is related at a conceptual level to a large number of lines of work that deal with optimization and uncertainty. But it differs from them mostly in either the active aspect or the optimization aspect. We'll discuss some of these related lines of work and the differences below. 

%(in theoretical computer science \citep{balkanski2016power}, machine learning \citep{balcan2006agnostic,balcan2007margin,castro2008minimax}, and artificial intelligence \citep{braziunas2006computational,blum2004preference}) that deal with optimization and uncertainty.
%However, most of them solve a different conceptual problem, failing to model either the optimization aspect or the active aspect. At a technical level, some works \citep{chen2014combinatorial, gabillon2016improved, chen2016combinatorial, chen2016pure, chen17nearly} on the \emph{combinatorial pure exploration} (CPE) problem, are relatively close to our model. But these works focus on combinatorial optimization problems, which is very different from the linear optimization problems investigated in this work. \sherry{Added more detailed comparison with the CPE problem here.}

\paragraph{Optimization under uncertainty}
Our problem considers optimization with uncertain model parameters, which is also a theme in stochastic programming~\citep{heyman2003stochastic,neely2010stochastic}, chance-constrained programming~\citep{ben2009robust}, and robust optimization~\citep{ben2009robust,bertsimas2004robust}. In both stochastic optimization and chance-constrained programming, there are no underlying true, fixed values of the parameters; instead, a probabilistic, distributional model of the parameters is used to capture the uncertainty and such model of uncertainty becomes part of the optimization formulation. Hence, optimality in expectation or approximate optimality is sought after under the probabilistic model. 
%In stochastic optimization, there are no underlying true, fixed constraints; instead constraints are modeled as having probabilistic uncertainty and the learner seeks a solution achieving optimality in expectation.
But in our problem underlying fixed parameters exist, and the problem is only how much information to gather about them.
Meanwhile, robust optimization models deterministic uncertainty (e.g. parameters come from a known set) and often seeks for a worst-case solution, \emph{i.e.} a solution that is feasible in the worst case over a set of possible constraints. 
A key distinction of our model is that there are underlying true values of the parameters and we do not incorporate any probabilistic or deterministic model of the uncertainty into the optimization problem itself. Instead, we take an "active querying" approach to approximately solve the true optimization problem with high probability. 

%We emphasize that a central feature of the model in this paper is \emph{noisy observations}: there are underlying true values of the parameters and the observations of the algorithm are only noisy samples of the parameters. The key challenge is to choose how many repeated samples of each parameter to draw. This ``active querying'' aspect of the problem is not to our knowledge present the stochastic programing, chance-constrained programming and robust optimization literature. 

%A key distinction is that in robust optimization there is no active stage to learn more about the constraints.
%The worst-case robust optimization solution could be very far from the true objective value depending on the actual constraints.
%In our case, we can do much better by progressively collecting information about the true constraints. But our setting allows some small tolerance on violating feasibility of the underlying true constraints (which will be discussed in the Problem Formulation section), while robust optimization almost always preserves feasibility. 

\paragraph{Artificial intelligence.}

%\yc{For the related work in AI, maybe we want to be explicit on what relates to the unknown c setting and what to the unknown b setting.}

Several existing lines of work in the artificial intelligence literature, deal with actively acquiring information about parameters of an optimization problem in order to solve it.
Preference elicitation~\citep{braziunas2006computational,blum2004preference} typically focuses on acquiring information about parameters of the \emph{objective} by querying a user about his preferences, this is similar to our goal for the unknown-$\V{c}$ setting. Relevant to our unknown-$\V{b}$ case, for more combinatorial problems, the constraint acquisition literature~\citep{oconnell2002strategies,bessiere2015constraint} is closer to our problem in some respects, as it posits an optimization problem with unknown constraints that must be learned via interactive querying. 
We emphasize that a central feature of the model in this paper is \emph{noisy observations}: the observations of the algorithm are only noisy samples of a true underlying parameter. The key challenge is to choose how many repeated samples of each parameter to draw. This aspect of the problem is not to our knowledge present in preference elicitation or model/constraint acquisition. 

\paragraph{Machine learning and theoretical computer science.}
Much work in \emph{active learning} considers acquiring data points iteratively with a goal of low sample complexity \citep{balcan2006agnostic,balcan2007margin,castro2008minimax}.The key difference to AIAO is between \emph{data} and \emph{parameters}.
In learning, the goal is to minimize the average or expectation of some loss function over a distribution of data points.
Other than its likelihood, each data point plays the same role in the problem.
Here, the focus is on how much information about each of various parameters is necessary to solve a structured optimization problem to a desired level of accuracy.
In other words, the key question here, which is how much an optimization algorithm needs to know about the parameters of the problem it is solving, does not apply in active learning.

A line of work on sample complexity of reinforcement learning (or approximate reinforcement learning)~\citep{kakade2003sample, lattimore2012pac, azar2012sample, DBLP:journals/corr/abs-1710-06100, chen2016stochastic} also bears some resemblance to our problem. A typical setting considered is solving a model-free Markov Decision Processes (MDP), where the transition probabilities and the reward functions are initially unknown but the algorithm can query an oracle to get samples. This problem is a special case of our AIALO problem with unknown $\V{A}$ and $\V{b}$ because an MDP can be formulated as an linear program. The solutions provided focuses on the particular structure of MDP, while we consider general linear programs. 

%This extended model will allow us to study a wider range of problems. Here we provide one possible application: the sample complexity problem in Reinforcement Learning \cite{kakade2003sample, DBLP:journals/corr/abs-1710-06100, chen2016stochastic}. The problem studies the sample complexity of learning the optimal policy when the transition matrix and the reward function of a Markov Decision Process are initially unknown. This problem is a special case of our AIALO problem with unknown $\V{A}$ and $\V{b}$ because any MDP can be formulated as an linear program, where the matrix $\V{A}$ and vector $\V{b}$ are unknown due to the uncertainty of the transition matrix and the reward function, as discussed in \cite{chen2016stochastic}.

Broadly related is recent work on \emph{optimization from samples} \citet{balkanski2016power}, which considers the sample complexity of a two-stage process: (1) draw some number of i.i.d. data points; (2) optimize some loss function or submodular function on the data.
In that setting, the algorithm sees a number of input-output pairs of the function, randomly distributed, and must eventually choose a particular input to optimize the function.
Therefore, it is quite different from our setting in both important ways: (1) the information collected are data points (and evaluations), as in ML above, rather than parameters as in our problem; (2) (so far) it is not active, but acquires information in a batch.

%\bo{TODO! Combinatorial pure exploration!}
A line of work that is closely related to our unknown-$\V{c}$ problem is the study of \emph{combinatorial pure exploration} (CPE) problem, where a learner collects samples of unknown parameters of an objective function to identify the optimal member in a solution set. The problem was first proposed in \cite{chen2014combinatorial}, and subsequently studied by~\cite{gabillon2016improved, chen2016combinatorial, chen2016pure, chen17nearly}.  CPE only considers combinatorial optimization problems whose solution set contains only binary vectors of length $n$. A recent work by~\cite{chen17nearly} extended CPE to a \emph{General-Sampling} problem by allowing general solution sets. Our unknown-$\V{c}$ problem can be fitted into the setting of General-Sampling. Our algorithm for unknown-$\V{c}$ was inspired by the work of~\cite{chen17nearly}, but leverages the structure of LP and hence has better sample complexity performance than directly treating it as a General-Sampling problem.
%, which is guaranteed no more than $\ln(1/\Delta)$ times the lower bound, whereas their algorithm can have an additional $n^2$ factor compared to the lower bound. 
The General-Sampling problem does not encompass all AIAO settings, e.g., our unknown-$\V{b}$ case.

%At a technical level, some works \citep{chen2014combinatorial, gabillon2016improved, chen2016combinatorial, chen2016pure, chen17nearly} on the \emph{combinatorial pure exploration} (CPE) problem, are relatively close to our model. But these works focus on combinatorial optimization problems, which is very different from the linear optimization problems investigated in this work. 

\section{MODEL AND PRELIMINARIES} \label{sec:prelim}

\subsection{THE AIALO PROBLEM}
We now formally define an instance $\mathcal{I}$ of the \emph{active information acquisition for linear optimization (AIALO)} problem.
We then describe the format of algorithms for solving this problem.
Note that one can easily extend this into a more general formal definition of AIAO, for more general optimization problems, but we leave this for future work.

An instance $\mathcal{I}$ consists of three components.
The first component consists of the \emph{parameters} of the underlying linear program on $n$ variables and $m$ constraints: a vector $\V{c} \in \R^n$, a vector $\V{b} \in \R^m$, and a matrix $\V{A} \in \R^{n\times m}$.
Naturally, these specify the program\footnote{Note that any linear program can be transformed into the given format with at most a factor $2$ increase in $n$ and $m$. }%\bo{cite or explain?}\yc{I think it's fine as it is.}}
\begin{equation}
  \max_{\V{x}} \V{c}^T \V{x} ~ \textrm{ s.t. } \V{A}\V{x} \le \V{b},  \V{x} \geq \V{0}.
\end{equation}
We assume for simplicity in this paper that all linear programs are feasible and are known \emph{a priori} to have a solution of norm at most $R$.
The second component specifies which parameters are \emph{initially known} and which are \emph{initially unknown}.
The third and final component specifies, for each unknown parameter (say $c_i$), of a $\sigma^2$-subgaussian distribution with mean equal to the value of the parameter.\footnote{Distribution $\mathcal{D}$ with mean $\mu$ is $\sigma^2$-subgaussian if, for $X \sim \mathcal{D}$, we have $\mathbb E[e^{t (X-\mu)}] \leq e^{\sigma^2 t^2/2}$ for all $t$. The family of sub-Gaussian distributions with parameter $\sigma$ encompasses all distributions that are supported on $[0,\sigma]$ as well as many unbounded distributions such as Gaussian distributions with variance $\sigma^2$.}

Given $\mathcal{I}$, we define the following sets of approximately-feasible, approximately-optimal solutions.
\begin{definition} \label{def:correct-solution}
  Given an instance $\mathcal{I}$, let $\V{x^*}$ be an optimal solution to the LP.
  Define $OPT(\mathcal{I};\varepsilon_1,\varepsilon_2)$ to be the set of solutions $\V{x}$ satisfying $\V{c}^T \V{x} \geq \V{c}^T \V{x^*} - \varepsilon_1$ and $\V{A}\V{x} \leq \V{b} + \varepsilon_2 \mathbf{1}$.
  We use $OPT(\mathcal{I})$ as shorthand for $OPT(\mathcal{I};0,0)$.
\end{definition}

\subsection{ALGORITHM SPECIFICATION}
An algorithm for the AIALO problem, run on an instance $\mathcal{I}$, functions as follows.
The algorithm is given as input $n$ (number of variables), $m$ (number of constraints), and $\sigma^2$ (subgaussian parameter).
It is also given the second component of the instance $\mathcal{I}$, i.e. a specification of which parameters are known and which are unknown.
For each parameter that is specified as ``known'', the algorithm is given the value of that parameter, e.g. it is given ``$A_{11}=42$.''
Finally, the algorithm is given an optimality parameter $\varepsilon_1$, a feasibility parameter $\varepsilon_2$, and a failure probability parameter $\delta$.
%\sherry{exchange the notation of optimality parameter and feasibility parameter, $\varepsilon_1$ and $\varepsilon_2$ here}

The algorithm may iteratively choose an unknown parameter and \emph{sample} that parameter: 
observe an independent and identically-distributed draw from the distribution corresponding to that parameter (as specified in the third component of the instance $\mathcal{I}$).
At some point, the algorithm stops and outputs a \emph{solution} $\V{x} \in \R^n$.

\begin{definition}[$(\delta,\varepsilon_1,\varepsilon_2)$-correct algorithm] \label{def:correct-alg}
  An algorithm $\mathcal{A}$ is \emph{$(\delta,\varepsilon_1,\varepsilon_2)$-correct} if for any instance $\mathcal{I}$ and inputs $(\delta,\varepsilon_1,\varepsilon_2)$, with probability at least $1-\delta$, $\mathcal{A}$ outputs a solution $\V{x} \in OPT(\mathcal{I}; \varepsilon_1,\varepsilon_2)$.
  In the case $\varepsilon_1=\varepsilon_2=0$, we say $\mathcal{A}$ is \emph{$\delta$-correct}.
\end{definition}

Our goal is to find algorithms with low \emph{sample complexity}, i.e., the number of samples drawn by the algorithm.

\ifodd 1
\newcommand{\rev}[1]{{\color{blue}#1}} %revise of the text
\newcommand{\com}[1]{} %comment of the text
\newcommand{\clar}[1]{\textbf{\color{green}(NEED CLARIFICATION: #1)}}
\newcommand{\response}[1]{\textbf{\color{magenta}(RESPONSE: #1)}} %response to comment
\else
\newcommand{\rev}[1]{#1}
\newcommand{\com}[1]{}
\newcommand{\clar}[1]{}
\newcommand{\response}[1]{}
\fi

\section{UNKNOWN CONSTRAINTS} \label{sec:problem}
We will first consider the \emph{unknown-$\V{b}$ case} where every parameter of the constraint vector $\V{b}$ is initially unknown, and all other parameters are initially known.
Geometrically, the algorithm is given an objective ``direction'' ($\V{c}$) and a set of constraint ``orientations'' ($\V{A}$), but does not initially know the ``offset'' or displacement $b_i$ of each constraint $i$.

In this setting, we do not expect to attain either exact feasibility or exact optimality, as the exact constraints can never be known, and in general an arbitrarily small change in constraints of an LP leads to a nonzero change in the value of the optimal solution.
%This represents a qualitative difference from the unknown-$\V{c}$ case.\footnote{In particular, it may be possible to hope for some sort of duality-based reduction from the unknown-$\V{b}$ case to the unknown-$\V{c}$ case, because the latter can be solved to exact feasibility and optimality; but the reverse seems less likely.}
    %This is an exciting direction for future work, especially in terms of duality relationships between approximate solutions.}

The section begins with a lower bound of the sample complexity (across all LP instances) of any correct algorithm. %Then we introduce UCB-Ellipsoid algorithm and give an upper bound of its sample complexity. Finally we compare our algorithm with the lower bound, as well as a naive approach using simulations. 

\begin{restatable}[Lower bound for unknown $\V{b}$]{theorem}{unknownbLower}
\label{lb2:unknownb}
Suppose we have a $(\delta, \varepsilon_1, \varepsilon_2)$-correct algorithm $\mathcal{A}$ where $\delta \in (0, 0.1), \varepsilon_1>0, \varepsilon_2>0$. Then for any $n>0$, there exists infinitely many instances of the AIALO problem with unknown-$\V{b}$ with $n$ variables with objective function $\Vert \V{c} \Vert_\infty = 1$ such that $\mathcal{A}$ must draw at least
\[
\Omega\left(n \ln(1/\delta)\cdot \max\{\varepsilon_1, \varepsilon_2\}^{-2}\right)
\]
samples in expectation on each of them.
\end{restatable}
The idea of the lower bound (proof in Appendix~\ref{app:unknownb_lower}) is that in the worst case, an algorithm must accurately estimate at least all $n$ binding constraints (in general with $n$ variables, up to $n$ constraints bind at the optimal solution). It remains open whether we can get a tighter lower bound which also captures the difficulty of ruling out non-binding constraints.

\subsection{ELLIPSOID-UCB ALGORITHM}
\paragraph{Background.}

%Before giving our algorithm, we describe 
%The high-level idea is to maintain a boundary, using an ``ellipsoid'' shape, around the space where the optimal solution must lie; then repeatedly cut this space using either the constraints or the objective.
The standard ellipsoid algorithm for linear programming begins with an ellipsoid $\mathcal{E}^{(0)}$ known to contain the optimal solution, 
%The algorithm 
%Formally, an ellipsoid is the set of points $\left\{\V{x} : \V{x}^T \V{P}^{-1} \V{x} \leq 1 \right\}$ specified by a positive definite matrix $\V{P}$.
%The algorithm sets $\V{x}^{(0)}$ to be the center of this ellipsoid.
Then, it checks two cases: (1) The center of this ellipsoid $\V{x}^{(0)}$ is feasible, or (2) it is not feasible.
If (2), say it violates constraint $i$, then the algorithm considers the halfspace defined by the constraint $\V{A}_i \V{x}^{(0)} \leq b_i$.
%\sherry{Actually we use halfspace $\V{A}_i \V{x} \le \V{A}_i \V{x}^{(0)}$ in our algorithm, since $b_i$ is unknown.}
If (1), the algorithm considers the halfspace defined by the ``constraint'' $\V{c}^T \V{x} \geq \V{c}^T \V{x}^{(0)}$, as the optimal solution must satisfy this constraint.
In either case, it updates to a new ellipsoid $\mathcal{E}^{(1)}$ defined as the minimal ellipsoid containing the intersection of $\mathcal{E}^{(0)}$ with the halfspace under consideration.

%Less formally, it ``cuts'' the current ellipsoid $\mathcal{E}^{(0)}$ into a smaller space using the violated constraint it found (or using the objective function if there is no violated constraint).
%It then takes a small ellipsoid $\mathcal{E}^{(1)}$ surrounding this space, which can be computed efficiently, and repeats until the volume of the ellipsoid is sufficiently small.
%Then it can terminate and return the best feasible $\V{x}^{(k)}$.
%Our modified ellipsoid algorithm (Algorithm \ref{alg:modified-ellipsoid}) closely follows this overall procedure.\yc{The beginning of this paragraph felt redundant given the previous paragraph. (The previous paragraph is fairly easy to follow and hence this less formal version doesn't add much.) Maybe combine these two paragraphs?}

%\paragraph{Bandit approaches.}
The obstacle is that, now, $\V{b}$ is initially unknown.
%When the ellipsoid algorithm needs to find a violated constraint (or certify that none is violated), how should we sample the constraints $b_i$?
A first observation is that we only need to find a single violated constraint, so there may be no need to sample most parameters at a given round.
A second observation is that it suffices to find the \textbf{most violated} constraint.
This can be beneficial as it may require only a few samples to find the most violated constraint; and in the event that no constraint is violated, we still need to find an upper bound on ``closest to violated'' constraint in order to certify that no constraint is violated.

To do so, we draw inspiration from algorithms for \emph{bandits} problems (whose details are not important to this paper).
Suppose we have $m$ distributions with means $\mu_1,\dots,\mu_m$ and variances $\sigma_1^2,\dots,\sigma_m^2$, and we wish to find the largest $\mu_i$.
After drawing a few samples from each distribution, we obtain estimates $\widehat{\mu}_i$ along with confidence intervals given by tail bounds.
Roughly, an ``upper confidence bound'' (UCB) algorithm (see \emph{e.g.} \cite{jamieson2014best}) for finding $\max_i \mu_i$ proceeds by always sampling the $i$ whose upper confidence bound is the highest.
% This bandit approach helps identify the best option (in our case the most violated constraints) quickly with high probability, without spending too much sampling on sub-optimal ones (non-crucial constraints).
%\yc{It will be nice to at least say what guarantees bandits algorithms, especially UCB, can provide.} \bo{I thought it's a distraction and space-waste to describe the whole bandits setting and results...}\yc{I wasn't thinking about describing the whole bandits setting and results (agree that that's not a good idea). I was thinking that it's nice to have for example a sentence saying that the bandits algorithm elegantly deals with the exploration and exploitation tradeoff or say something like it has good regret bound. The reason is the ``why bandits'' question. Why they are good choice for our problem? Having said that, it's ok to leave it as it is if we don't have space.}
% The idea behind these algorithms is to manage exploration and exploitation: by picking the highest upper bound, the algorithm is sure to either be sampling from a an option that actually is large, or else from one that is quite uncertain.%\bo{added bandits sentence}

%In our case, at each stage of the ellipsoid algorithm, we wish to find the largest $\V{A}_i \V{x} - b_i$ or ensure that all are nonpositive.
We therefore will propose a UCB-style approach to doing so, but with the advantage that we can re-use any samples from earlier stages of the ellipsoid algorithm.

\paragraph{Algorithm and results.}
%\bo{Removed term ``mini-stage'', it's a nice term, but unfortunately I think doesn't fit well enough here. the issue is Alg 2 might return a value without even sampling...}
Ellipsoid-UCB is given in Algorithm \ref{alg:modified-ellipsoid}.
At each round $k=1,2,\dots$, we choose the center point $\V{x}^{(k)}$ of the current ellipsoid $\mathcal{E}^{(k)}$ and call the subroutine Algorithm \ref{alg:ucb-method} to draw samples and check for violated constraints.
We use the result of the oracle to cut the current space exactly as in the standard ellipsoid method, and continue.

Some notation: $t$ is used to track the total number of samples drawn (from all parameters) and $T_i(t)$ denotes the number of samples of $b_i$ drawn up to ``time'' $t$.
The average of these samples is:
\begin{definition} \label{def:b-hat}
  Let $X_{i,s}$ denote the $s$-th sample of $b_i$ and let $T_i(t)$ denote the number of times $b_i$ is sampled in the first $t$ samples. Define $\widehat{b}_{i,T_i(t)}=\sum_{s=1}^{T_i(t)} X_{i,s}/T_i(t)$ to be the empirical mean of $b_i$ up to ``time'' $t$.
\end{definition}

\begin{algorithm}[!h]                      % enter the algorithm environment
\small
\caption{Modified ellipsoid algorithm}
\label{alg:modified-ellipsoid}
\begin{algorithmic}
    \STATE Let $\mathcal{E}^{(0)}$ be the initial ellipsoid containing the feasible region.\\
    \STATE Draw one sample for each $b_i$, $i\in [m]$.\\
    \STATE Let $k = 0$ and $t=m$.\\
    \STATE Let $T_i(t) = 1$ for all $i$. \\%\bo{added this line}\\
    \WHILE{stopping criterion is not met\footnotemark}
            \STATE Let $\V{x}^{(k)}$ be the center of $\mathcal{E}^{(k)}$
            \STATE Call UCB method to get constraint $i$ or ``feasible''
            \IF{$\V{x}^{(k)}$ is feasible}
%                \STATE Let $\V{x} \gets \V{x}^{(k)}$ if $\V{x}^{(k)}$ is the first feasible point, or $\V{x}^{(k)}$ is better, i.e. $\V{c}^T \V{x}^{(k)} > \V{c}^T \V{x}$. \yc{We need to initialize $\V{x}$ earlier on, I think.} \sherry{change the condition of update}
                \STATE $\V{x} \gets \V{x}^{(k)}$ if $\V{x}$ not initialized or $\V{c}^T \V{x}^{(k)} > \V{c}^T \V{x}$. %\bo{Updated this line.}
%                \STATE Compare $\Sol$ with $\V{x}^{(k)}$ \bo{$\Sol \gets \V{x}^{(k)}$ Correct?} \sherry{I think we cannot replace $\Sol$ with $\V{x}^{(k)}$ directly cause ellipsoid does not completely remove points $c^T x \le c^T x^{(k)}$}
                \STATE $\V{y} \gets -\V{c}$
            \ELSE
                \STATE $\V{y}\gets \V{A}_{i}^T$
            \ENDIF
            \STATE Let $\mathcal{E}^{(k+1)}$ be the minimal ellipsoid that contains $\mathcal{E}^{(k)} \cap \{\V{p} ~:~ \V{y}^T \V{p} \le \V{y}^T\V{x}^{(k)}\}$\\
            \STATE Let $k \gets k+1$
    \ENDWHILE
    \STATE Output $\Sol$ or ``failure'' if it was never set.
\end{algorithmic}
\end{algorithm}
\footnotetext{Our stopping criterion is exactly the same as in the standard ellipsoid algorithm, for which there are a variety of possible criteria that work. In particular, one is $\sqrt{\V{c}^T \V{P}^{-1} \V{c}} \leq \min\{\varepsilon_1,\varepsilon_2\}$, where $P$ is the matrix corresponding to ellipsoid $\mathcal{E}^{(k)}$ as discussed above.}

\begin{algorithm}[!h]                      % enter the algorithm environment
\small
\caption{UCB-method}
\label{alg:ucb-method}
\begin{algorithmic}
    \STATE Input $\V{x}^{(k)}$
    \STATE Output either index $j$ of a violated constraint, or ``feasible''.
    \STATE Set $\delta' = \left(\frac{\delta}{20m}\right)^{2/3}$
    \LOOP
            \STATE  1. Let $j$ be the constraint with the largest index,
            \[
                j = \arg\max_i \V{A}_i \V{x}^{(k)} - \widehat{b}_{i,T_i(t)} + U_i(T_i(t)),
            \]
            where $U_i(s) = 3\sqrt{2\sigma_i^2\log\left(\log(3s/2)/\delta'\right)/s}$ and $\widehat{b}_{i,T_i(t)}$ as in Definition \ref{def:b-hat}.
            \STATE 2. If $\V{A}_j \V{x}^{(k)} - \widehat{b}_{j,T_j(t)} - U_j(T_j(t))> 0$ return $j$.
            \STATE 3. If $\V{A}_j \V{x}^{(k)} - \widehat{b}_{j,T_j(t)} + U_j(T_j(t))< 0$ return ``feasible''.
            \STATE 4. If $U_j(T_j(t))< \varepsilon_2/2$ return ``feasible''.
            \STATE 5. Let $t \gets t+1$
            \STATE 6. Draw a sample of $b_j$.
            \STATE 7. Let $T_j(t) = T_j(t-1) + 1$.
            \STATE 8. Let $T_i(t) = T_i(t-1)$ for all $i\neq j$.
    \ENDLOOP
\end{algorithmic}
\end{algorithm}

The performance of the algorithm is measured by how many samples (observations) it needs to draw.
To state our theoretical results, define $V_i(k) = \V{A}_i \V{x}^{(k)} - b_i$ to be the amount by which the $i$-th constraint is violated by $\V{x}^{(k)}$, and $V^*(k) = \max_i V_i(k)$. Define $gap_{i,\varepsilon}(k) = \max\{|V_i(k)|,  V^*(k) -V_i(k), \varepsilon\}$ and $\Delta_{i,\varepsilon}$ $=  \min_k gap_{i,\varepsilon}(k)$. 

\begin{restatable}[Ellipsoid-UCB algorithm]{theorem}{sampleEllipsoidUCB}
\label{main}
 The Ellipsoid-UCB algorithm is $(\delta,\epsilon_1,\epsilon_2)$-correct and with probability $1-\delta$, draws at most the following number of samples: 
 {\small
 \[ 
O\left(\sum_{i=1}^m \frac{\sigma_i^2}{\Delta_{i,\varepsilon_2/2}^2}\log\frac{m}{\delta} + \sum_{i=1}^m \frac{\sigma_i^2}{\Delta_{i,\varepsilon_2/2}^2}\log\log\left(\frac{\sigma_i^2}{\Delta_{i,\varepsilon_2/2}^2}\right)\right).
\]}
Specifically, the number of samples used for $b_i$ is at most $\frac{\sigma_i^2}{\Delta_{i,\varepsilon_2/2}^2}\left(\log(m/\delta)+\log\log(\sigma_i^2/\Delta_{i,\varepsilon_2/2}^2)\right)$.
\end{restatable}
\begin{proof}[Proof Sketch:]
%Our analysis is inspired by the techniques used in \cite{jamieson2014lil}.
Define event $\mathcal{A}$ to be the event that $\left|\B_{i,s} - b_i \right| \le U_i(s)$ for all $s$ and $i\in[m]$.
%\yc{Maybe use $s$ instead of $t$ here since we've been using $t$ for mini-stages.} \sherry{Use $s$ instead} 
According to Lemma 3 in \cite{jamieson2014lil}, $\mathcal{A}$ holds with probability at least $1-\delta$. %We prove the correctness and the sample complexity conditioning on that $\mathcal{A}$ holds.

\textbf{Correctness:} %If UCB-method always gives a correct answer, the ellipsoid algorithm will be able to find an $\varepsilon_1$-suboptimal solution. So we only need to prove the correctness of the UCB-method. 
Conditioning on event $\mathcal{A}$ holds, our UCB method will only return a constraint that is violated (line 2) and when it returns ``feasible'', no constraint is violated more than $\varepsilon_2$ (line 3 and 4).

\textbf{Number of samples:} We bound the number of samples used on each constraint separately. Consider a fixed ellipsoid iteration $k$ in which UCB method is given input $\V{x}^{(k)}$, the key idea is to prove that if $b_i$ is sampled in this iteration at ``time'' $t$, $U_i(T_i(t))$ should be larger than $gap_{i,\varepsilon_2/2}(k)$. This gives an upper bound of $T_i(t)$. %Thus for each ellipsoid iteration $k$, we have an upper bound of $T_i(t)$. 
Taking the maximum of them, we get the final result.

\end{proof}

%
%We make several notes on the bounds.
%\begin{itemize}
%\item The bound depends on XXX.
%\item The intuition for our results is that at step stage of Ellipsoid, we only need to track the most likely to become violated constraint. XXX. This is in analogy to selecting the ``best arm'' in the bandit setting.  This helps us greatly reduce sampling complexity for less important constraints. And this adaptive procedure grants us the opportunity to tell the useful constraints from others as early as possible.
%\end{itemize}

\paragraph{Discussion.}
To understand the bound, suppose for simplicity that each $\sigma_i = 1$.
We observe that the first term will dominate in all reasonable parameter settings, so we can ignore the second summation in this discussion.

Next, note that each term in the sum reflects a bound on how many times constraint $i$ must be sampled over the course of the algorithm.
This depends inversely on $\Delta_{i,\varepsilon_2/2}$, which is a minimum over all stages $k$ of the ``precision'' we need of constraint $i$ at stage $k$.
We only need a very precise estimate if both of the following conditions are satisfied:
(1) $|V_i(k)|$ is small, meaning that the ellipsoid center $\V{x}^{(k)}$ is very close to binding constraint $i$;
 (2) There is no other constraint that is significantly violated, meaning that $i$ is very close to the most-violated constraint for $\V{x}^{(k)}$ if any.
Because this is unlikely to happen for most constraints, we expect $\Delta_{i,\varepsilon_2/2}$ to generally be large (leading to a good bound), although we do not have more precise theoretical bounds.
The only constraints where we might expect $\Delta_{i,\varepsilon_2/2}$ to be small are the binding constraints, which we expect to come close to satisfying the above two conditions at some point.
Indeed, this seems inevitable for any algorithm, as we explore in our experiments.

\paragraph{Comparison to static approach.}
Again suppose each $\sigma_i = 1$ for simplicity.
Note that each $\Delta_{i,\varepsilon_2/2} \geq \varepsilon_2/2$.
This implies that our bound is always better than $O\left(\frac{m \log(m/\delta)}{\varepsilon_2^2}\right)$, ignoring the dominated second term.

The static approach is to measure each $b_i$ with enough samples to obtain a good precision so that relaxed feasibility can be satisfied with high probability, then solve the linear program using the estimated constraints.
This uses $\frac{4m\log(m/\delta)}{\varepsilon_2^2}$ samples.
(This number comes from using tail bounds to ensure good precision is achieved on every $b_i$.)

Therefore, the UCB-Ellipsoid algorithm dominates the static approach up to some constant factors and can show dramatic instance-wise improvements.
Indeed, in some simple cases, such as the number of variables equal to the number of constraints, we do not expect any algorithm to be able to improve over the static approach.
However, a nice direction for future work is to show that, if $m$ is very large compared to $n$, then the UCB-Ellipsoid algorithm (or some other algorithm) is guaranteed to asymptotically improve on the static approach.

%\bo{commented out comparison to CLUCB. I think for space it should either be removed or go to the appendix.}
%\paragraph{A shortest path algorithm with bandit feedback \& lower bound.}
%We note the CLUCB algorithm proposed in \cite{chen2014combinatorial} can be applied to actively collect data for shortest path problem, which can be formulated as an LP:
%\begin{align*}
%\max_d ~~& d_t - d_s,\\
% \textrm{s.t.}  ~~& d_v - d_u \le l_{uv}, \textrm{ for all edge }(u,v),
% \end{align*}
% where $l_{uv}$ is the length of the edge. The motivation and goal of our work is different from the above reference, which focuses on a class of combinatorial problems; and our results apply to a general LP formulation. Nevertheless our results share certain similarity with theirs. Particularly the number of samples CLUCB algorithm requires to find the shortest path with probability $1-\delta$ is bounded by $O(\sigma^2 n^2 H \log(m H/\delta))$, where $H=\sum_e \Delta_e^{-2}$, $\Delta_e$ is the length difference between the shortest path that goes through $e$ and the shortest path that does not go through $e$. This quantity $\Delta_e$ plays a very similar role as $\Delta_{i,\varepsilon_2/2}$ in our bound. Making a rigorous connection between these two quantities and bounds is a very interesting question for future study. The lower bound proved in \cite{chen2014combinatorial}, which is $\Omega(H\log(1/\delta))$ may be able to help establish a lower bound on sample complexity for our setting.

%\sherry{Place an intuitive lower bound here?}}

%\com{matching the lower bound?}

\section{UNKNOWN OBJECTIVE FUNCTION} \label{sec:unknownc}
In this section, we consider the \emph{unknown-$\V{c}$ case}.
Here, every parameter of the objective $\V{c}$ is initially unknown, and all other parameters are initially known.
Geometrically, the algorithm is initially given an exact description of the feasible polytope, in the form of $\V{A}\V{x} \leq \V{b}$ and $\V{x} \geq 0$, but no information about the ``direction'' of the objective.

Because the constraints are known exactly, we focus on exact feasibility in this setting, i.e. $\varepsilon_2 = 0$.
We also focus on an information-theoretic understanding of the problem, and produce an essentially optimal but computationally inefficient algorithm. We assume that there is a unique optimal solution $\V{x^*}, $\footnote{If we only aim for a $\varepsilon_1$-suboptimal solution, we can terminate our algorithm when $\varepsilon^{(r)}$ (defined in Line 5 of Algorithm~\ref{unknowncAlgo}) becomes smaller than $\varepsilon_1/2$, such that the algorithm no longer requires the best point to be unique.} and consider the problem of finding an exact optimal solution with confidence $\delta$ (i.e., a $\delta$-correct algorithm). We also make the simplifying assumption that each parameter's distribution is a Gaussian of variance $1$ (in particular is $1$-subgaussian). Our results can be easily extend to the general case.

\sherry{Added a line here.} Our approaches are based on the techniques used in \cite{chen17nearly}, but address a different class of optimization problems. We thus use the same notations as in \cite{chen17nearly}. We first introduce a function $Low(\mathcal{I})$ that characterizes the sample complexity required for an LP instance $\mathcal{I}$.  The function $Low(\mathcal{I})$ is defined by the solution of a convex program. We then give an instance-wise lower bound in terms of the $Low(\mathcal{I})$ function and the failure probability parameter $\delta$. We also formulate a worst-case lower bound of the problem, which is polynomially related to the instance-wise lower bound.
Finally, we give an algorithm based on successive elimination that matches the worst-case lower bound within a factor of $\ln(1/\Delta)$, where $\Delta$ is the gap between the objective function value of the optimal extreme point ($\V{x^*}$) and the second-best.

\subsection{LOWER BOUNDS}

The function $Low(\mathcal{I})$ is defined as follows.
\begin{definition}[$Low(\mathcal{I})$]
  For any instance $\mathcal{I}$ of AIALO (or more generally, for any linear program), we define $Low(\mathcal{I}) \in \R$ to be the optimal solution to the following convex program.
  		\begin{align}
           \min_\tau \quad & \sum_{i=1}^n \tau_i \label{low}\\
          s.t.\quad & \sum_{i=1}^n \frac{(s_i^{(k)} -x_i^*)^2}{\tau_i} \le \left( \V{c}^T (\V{x}^* - \V{s}^{(k)}) \right)^2, \forall k \notag \\
           &\tau_i\ge 0, \forall i \notag
          \end{align}
  Here $\V{x}^*$ is the optimal solution to the LP in $\mathcal{I}$ and $\V{s}^{(1)},\dots, \V{s}^{(k)}$ are the extreme points of the feasible region $\{\V{x}:\V{A}\V{x}\le \V{b}, \V{x}\ge \V{0}\}$.  % \footnote{It is actually not necessary to consider all the corner points. The theorem still holds if we define $\V{x}^{(1)},\dots, \V{x}^{(k)}$ to be $\V{x}^*$'s neighboring corner points on feasible region. We include all the corner points only for simplicity in the proof of Theorem~\ref{unknownc}}
\end{definition}

For intuition about $Low(\mathcal{I})$, consider a thought experiment where we are given %the objective $\V{c}$ and
 an extreme point $\V{x^*}$, and we want to check whether or not $\Sol^*$ is the optimal solution using as few samples as possible.
Given our empirical estimate $\widehat{\V{c}}$ we would like to have enough samples so that with high probability, for each $\V{s}^{(k)} \neq \Sol^*$, we have
  \[  \widehat{\V{c}}^\transpose(\Sol^* - \V{s}^{(k)}) > 0  ~~ \iff ~~ \V{c}^\transpose(\Sol^* - \V{s}^{(k)}) > 0 . \]
This will hold by a standard concentration bound (Lemma~\ref{gaussianLem}) \emph{if enough samples of each parameter are drawn}; in particular, ``enough'' is given by the $k$-th constraint in (\ref{low}).
 
 \begin{restatable}[Instance lower bound]{theorem}{unknowncInstLowerbound} \label{thm:instance_lowerbound}
        Let $\mathcal{I}$ be an instance of AIALO in the unknown-$\V{c}$ case. For $0<\delta<0.1$, any $\delta$-correct algorithm $\mathcal{A}$  must draw  $$\Omega(Low(\mathcal{I}) \ln \delta^{-1})$$ samples in expectation on $\mathcal{I}$.
 \end{restatable}

We believe that it is unlikely for an algorithm to match the instance-wise lower bound without knowing the value of $\V{c}$ and $\Sol^*$ in the definition of $Low(\mathcal{I})$. To formally prove this claim, for any $\delta$-correct algorithm $\mathcal{A}$, we construct a group of LP instances that share the same feasible region $\V{A} \V{x} \le \V{b}$, $\V{x}\ge \V{0}$ but have different objective functions and different optimal solutions. We prove that $\mathcal{A}$ will have unmatched performance on at least one of these LP instances.

Our worst-case lower bound can be stated as follows.
\begin{restatable}[Worst-case lower bound for unknown $\V{c}$]{theorem}{unknowncWorstLowerbound}
\label{thm:worst_lowerbound}
        Let $n$ be a positive integer and $\delta\in (0, 0.1)$. For any $\delta$-correct algorithm $\mathcal{A}$, there exists an infinite sequence of LP instances with $n$ variables, $\mathcal{I}_1, \mathcal{I}_2, \dots$, such that $\mathcal{A}$ takes
        \[
        \Omega\left(Low(\mathcal{I}_k)(\ln |S^{(1)}_k| + \ln \delta^{-1}) \right)
        \]
        samples in expectation on $\mathcal{I}_k$, where $S^{(1)}_k$ is the set of all extreme points of the feasible region of $\mathcal{I}_k$, and $Low(\mathcal{I}_k)$ goes to infinity.
        \end{restatable}
%We postpone the proof of the worst-case lower bound to Appendix~\ref{app:unknownc_worst}.

\subsection{SUCCESSIVE ELIMINATION ALGORITHM}
    Before the description of the algorithm, we first define a function $LowAll(S, \varepsilon, \delta)$ that indicates  the number of samples we should take for each $c_i$, such that the difference in objective value between any two points in $S$ can be estimated to an accuracy $\varepsilon$ with probability $1-\delta$. Define $LowAll(S, \varepsilon, \delta)$ to be the optimal solution of the following convex program,
        \begin{align}
        \min_{\tau} \quad & \sum_{i=1}^n \tau_i \label{lowall}\\
        s.t.\quad & \sum_{i=1}^n \frac{(x_i -y_i)^2}{\tau_i} \le \frac{\varepsilon^2}{2 \ln(2/\delta)}, \forall \V{x}, \V{y} \in S \notag\\
         &\tau_i\ge 0, \forall i. \notag
        \end{align}
    Our algorithm starts with a set $S^{(1)}$ that contains all extreme points of the feasible region $\{\V{x}: \V{A}\V{x}\le \V{b}, \V{x}\ge \V{0}\}$, which is the set of all possible optimal solutions. We first draw samples so that the difference between each pairs in $S^{(1)}$ is estimated to accuracy $\varepsilon^{(1)}$.  Then we delete all points that are not optimal with high probability. In the next iteration, we halve the accuracy $\varepsilon^{(2)} = \varepsilon^{(1)}/2$ and repeat the process. The algorithm terminates when the set contains only one point.
 	\begin{algorithm}[H]                      % enter the algorithm environment
	\small
                \caption{A successive elimination algorithm}
                \label{unknowncAlgo}
                \begin{algorithmic}[1]
                    \STATE $S^{(1)}\gets$ set of all extreme points of feasible region $\{\V{x}: \V{A}\V{x}\le \V{b}, \V{x}\ge \V{0}\}$
                    \STATE $r\gets 1$
                    \STATE $\lambda\gets 10$
                    \WHILE{$|S^{(r)}|>1$}
                        \STATE $\varepsilon^{(r)} \gets 2^{-r}$, $\delta^{(r)} \gets \delta/(10 r^2 |S^{(1)}|^2)$
                        \STATE $(t^{(r)}_1, \dots, t^{(r)}_n) \gets LowAll(S^{(r)}, \varepsilon^{(r)}/\lambda, \delta^{(r)})$
                        \STATE Sample $c_i$ for $t^{(r)}_i$ times. Let $\widehat{c}^{(r)}_i$ be its empirical mean
                        \STATE Let $\V{x}^{(r)}$ be the optimal solution in $S^{(r)}$ with respect to $\widehat{\V{c}}^{(r)}$
                        \STATE Eliminate the points in $S^{(r)}$ that are $\varepsilon^{(r)}/2+2\varepsilon^{(r)}/\lambda$ worse than $\V{x}^{(r)}$ when the objective function is $ \widehat{\V{c}}^{(r)}$,
                        \begin{align} \label{eliminateCond}
                        	S^{(r+1)} \gets & \{ \V{x}\in S^{(r)}: \langle \V{x}, \widehat{\V{c}}^{(r)} \rangle \notag \\
						& \ge \langle \V{x}^{(r)},  \widehat{\V{c}}^{(r)} \rangle-\varepsilon^{(r)}/2-2\varepsilon^{(r)}/\lambda\}
			\end{align}
                        \STATE $r \gets r+1$
                    \ENDWHILE
                    \STATE Output $\V{x} \in S^{(r)}$
                \end{algorithmic}
                \end{algorithm}
The algorithm has the following sample complexity bound.
    \begin{restatable}[Sample complexity of Algorithm~\ref{unknowncAlgo}]{theorem}{unknowncSample}
     \label{thm:unknownc}
    For the AIALO with unknown-$\V{c}$ problem, Algorithm \ref{unknowncAlgo} is $\delta$-correct and, on instance $\mathcal{I}$, with probability $1-\delta$ draws at most the following number of samples:
     \[
        O\left( Low(\mathcal{I}) \ln \Delta^{-1} ( \ln |S^{(1)}| + \ln \delta^{-1} + \ln \ln \Delta^{-1})\right),
     \]
     where $S^{(1)}$ is the set of all extreme points of the feasible region and $\Delta$ is the gap in objective value between the optimal extreme point and the second-best,
     \[
     	\Delta = \max_{\Sol \in S^{(1)}} \V{c}^T \V{x} - \max_{\V{x} \in S^{(1)} \setminus \Sol^*} \V{c}^T \V{x} .
     \]
    \end{restatable}
\begin{proof}[Proof Sketch:]
%  Define event $\mathcal{E}$ to be the event that for each iteration $r$ and each pair of extreme points in $S^{(r)}$, the difference between them is estimated within error $\varepsilon^{(r)}/\lambda$.
  %More specifically, $|(\V{x}-\V{y})^T (\widehat{\V{c}}^{(r) - \V{c}))|\le \varepsilon^{(r) /\lambda$.
  We prove that conditioning on a good event $\mathcal{E}$ that holds with probability at least $1-\delta$, the algorithm will not delete the optimal solution and will terminate before $\lfloor \log(\Delta^{-1})\rfloor + 1$ iterations. Then we bound the number of samples used in  iteration $r$ by showing that the optimal solution of $Low(\mathcal{I})$ times $\alpha^{(r)} = 32 \lambda^2 \ln(2/\delta^{(r)})$ is a feasible solution of the convex program that defines $LowAll(S^{(r)}, \varepsilon^{(r)}/\lambda, \delta^{(r)})$. Therefore the number of samples used in iteration $r$ is no more than $\alpha^{(r)} Low(\mathcal{I})$.
\end{proof}
This matches the worst-case lower bound within a problem-dependent factor $\ln(1/\Delta)$. 
Notice, however, that the size of $|S^{(1)}|$ can be exponentially large, and so is the size of the convex program~\eqref{lowall}. 
So Algorithm \ref{unknowncAlgo} is computationally inefficient if implemented straightforwardly, and it remains open whether the algorithm can be implemented in polynomial time or an alternative algorithm with similar sample complexity and better performance can be found.

\section{EXPERIMENTS}

In this section, we investigate the empirical number of samples used by Ellipsoid-UCB algorithm for the unknown-$\V{b}$ case of AIALO. We fix $\delta = 0.1$ and focus on the impact of the other parameters\footnote{$99.5$ percent of the outputs turn out to satisfy relaxed feasibility and relaxed optimality.}, which are more interesting.

 We compare three algorithms on randomly generated LP problems.
The first is Ellipsoid-UCB.
The second is the naive ``static approach'', namely, draw $4\sigma^2\log(m/\delta)/\varepsilon_2^2$ samples of each constraint, then solve the LP using estimated means of the parameters.
      (This is the same approach mentioned in the previous section, except that previously we discussed the case $\sigma=1$ for simplicity.)
The third is designed to intuitively match the lower bound of Theorem \ref{lb2:unknownb}: Draw $4\sigma^2\log(d/\delta)/\varepsilon_2^2$ samples of each of only the binding constraints, where there are $d$ of them, then solve the LP using estimated means of the $b_i$.
      (For a more fair comparison, we use the same tail bound to derive the number of samples needed for high confidence, so that the constants match more appropriately.)

 We generate instances as follows. $\V{c}$ is sampled from $[-10, 10]^n$ uniformly at random. $\V{b}$ is uniformly drawn from $[0,10]^n$. Each $\V{A}_i$ is sampled from unit ball uniformly at random. Notice that the choice of $b_i\ge 0$ guarantees feasibility because the origin is always a feasible solution. We also add additional constraints $x_i \le 500$ to make sure that the LP generated is bounded. When the algorithm makes an observation, a sample is drawn from Gaussian distribution with variance $\sigma^2$.

  In Figure \ref{fig:compare-static}, each algorithm's number of samples (average of $50$ instances) is plotted as function of different parameters. The number of samples used by Ellipsoid-UCB is proportional to $n$, $\sigma^2$ and $\varepsilon^{-2}$. However, it does not change much as $m$ increases.\footnote{Indeed, the standard ellipsoid algorithm for linear programming requires a number of iterations that is bounded in terms of the number of variables regardless of the number of constraints.} This will not be surprising if ellipsoid uses most of its samples on binding constraints, just as the lower bound does.
  This is shown in Table \ref{table}, where it can be seen that Ellipsoid-UCB requires much fewer samples of non-binding constraints than binding constraints.
\begin{figure}[!h]
\subfigure[]{
\label{fig2a} %% label for first subfigure
\includegraphics[width=0.48\linewidth, height = 0.48\linewidth]{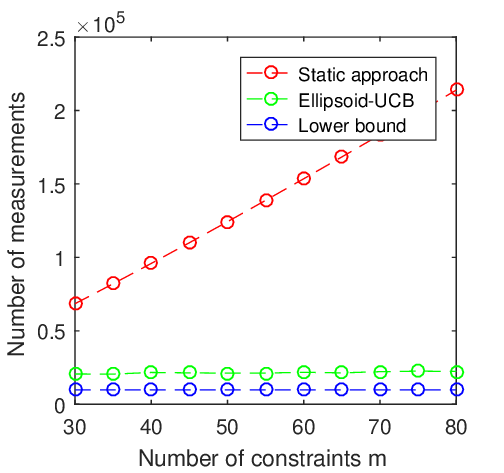}}
\subfigure[]{
\label{fig2b} %% label for second subfigure
\includegraphics[width=0.48\linewidth, height = 0.48\linewidth]{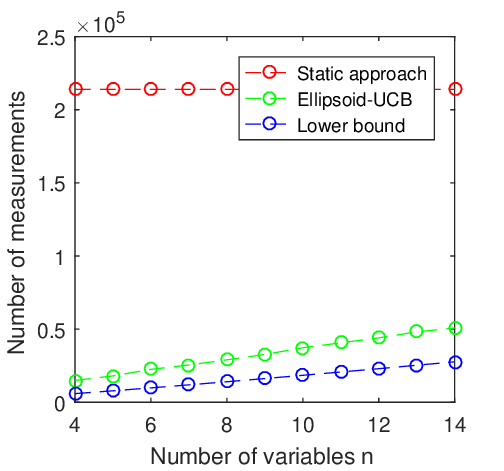}}

\subfigure[]{
\label{fig2c} %% label for second subfigure
\includegraphics[width=0.48\linewidth, height = 0.48\linewidth]{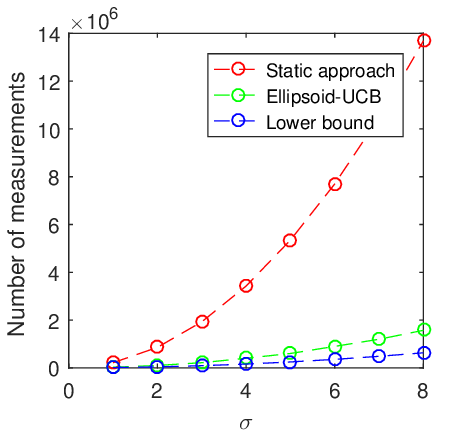}}
\subfigure[]{
\label{fig2d} %% label for second subfigure
\includegraphics[width=0.48\linewidth, height = 0.48\linewidth]{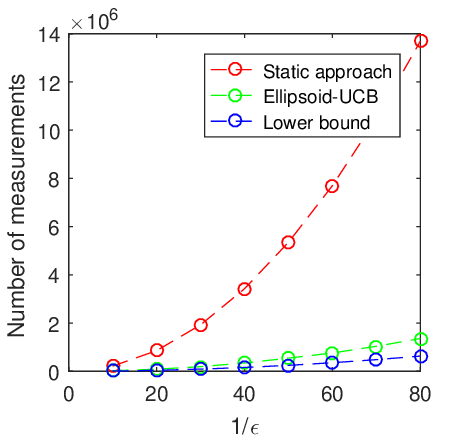}}

\caption{\small Number of samples as we vary $m$, $n$, $\sigma$ and $1/\varepsilon$. Every data point is the mean of 50 randomly drawn problem instances. The baseline parameters are $m=80$, $n=6$, $\sigma =1$, $\varepsilon_1 = \varepsilon_2 = 0.1$. In figure (d), $\varepsilon_1 =\varepsilon_2 =\varepsilon$. }
\label{fig:compare-static} %% label for entire figure
\end{figure}

\begin{table}[!h]
\centering
\small
    \begin{tabular}{|c|c|c|}
    \hline
     &Binding & Non-binding\\
    \hline
    Static approach & 2674 & 2674 \\
    \hline
    Ellipsoid-UCB    &  3325 & 11.7\\
    \hline
    Lower bound & 1476 & 0\\
    \hline
\end{tabular}

\caption{\small Average number of samples used per binding constraint and per non-binding constraint. Numbers are average from $100$ trials. Here, $m=80$, $n=4$, $\sigma=1$ $\varepsilon_1 =\varepsilon_2 = 0.1$.}

\label{table}
\end{table}

 Figure \ref{fig:cdfs} addresses the variance in the number of samples drawn by Ellipsoid-UCB by plotting its empirical CDF over $500$ random trials. The horizontal axis is the ratio of samples required by Ellipsoid-UCB to those of the lower bound.
For comparison, we also mention $R$, the ratio between the performances of the static approach and the lower bound.
These results suggest that the variance is quite moderate, particularly when the total number of samples needed grows.

%\yc{If (and only if) we have time, move the legend boxes in Figure 1 so that they don't stack on the actual lines. This can be done at camera ready stage.}
%\yc{Nothing need to be changed. I'm just wondering whether there is any specific reason that we used 100 trials for this table (and different parameter values from Figure 1) and 50 trials for all others.}

\begin{figure}[!h]
\centering
\includegraphics[width=0.5\textwidth]{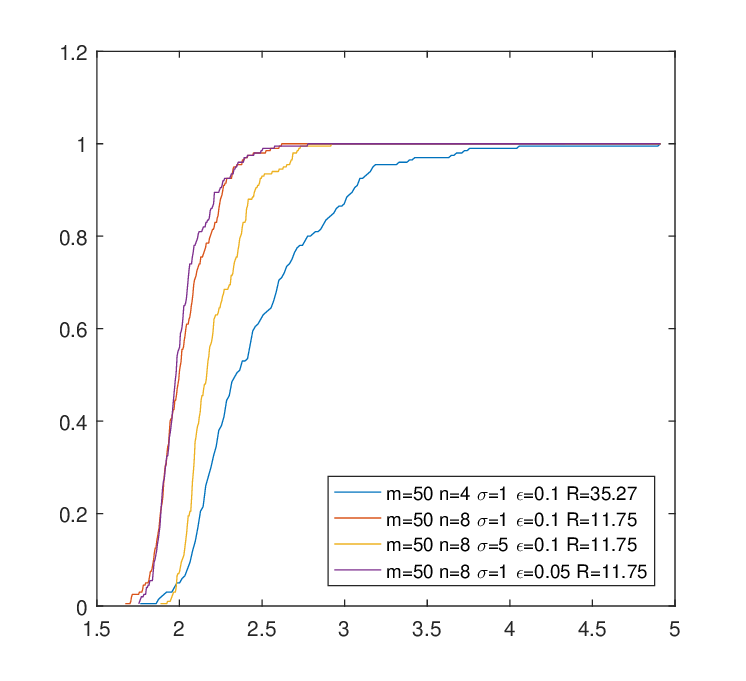}
\caption{\small Empirical cumulative distribution function of Ellipsoid-UCB's number of samples, in units of the ``lower bound'', over $500$ trials. Note that the lower bound varies when parameters change. $R = \frac{m\log(m)}{d\log(d)}$ is the ratio between the number of samples used by static approach and lower bound.}
\label{fig:cdfs}
\end{figure}

\section{DISCUSSION AND FUTURE WORK} \label{sec:discussion}
%\yc{Use vector notations for $x$, $A$, and $b$ for consistency.}
\sherry{Added the sample complexity of RL here}
One question is whether we can extend our results to situations when the constraint matrix $\V{A}$ is unknown as well as $\V{b}$. The goal is again to solve the problem with few total observations. This extended model will allow us to study a wider range of problems. For example, the sample complexity problem in Reinforcement Learning studied by \cite{kakade2003sample, DBLP:journals/corr/abs-1710-06100, chen2016stochastic} is a special case of our AIALO problem with unknown $\V{A}$ and $\V{b}$.

A second extension to the model would allow algorithms access to varying qualities of samples for varying costs.
For instance, perhaps some crowd workers can give very low-variance estimates for high costs, while some workers can give cheaper estimates, but have larger variance.
In this case, some preliminary theoretical investigations suggest picking the worker that minimizes the product (price)(variance).
A direction for future work is for the algorithm to select samples dynamically depending on the payment-variance tradeoffs currently available.
A final interested direction is a more mechanism-design approach where the designer collects bids from the agents and selects a winner whose data is used to update the algorithm.

\subsubsection*{Acknowledgements}
The authors are grateful to anonymous reviewers for their thorough reviews. This work is supported in part by National Science Foundation under grant CCF-1718549 and the Harvard Data Science Initiative. 

\vfill
\break

\bibliographystyle{unsrtnat}
\bibliography{references}
\vfill
\break

\appendix

\section{Motivating Examples} \label{app:motivate}
Our settings can be used to model many real-world optimization problems. In this section, we expand on some example real-world problems that fall into our framework:
\begin{itemize}
\item[(i)] A company who wants to decide a  production plan to maximize profit faces a linear program. But when entering a new market, the company may not initially know the average unit price $c_i$ for their different products in this market. Sampling $c_i$ corresponds to surveying a consumer on his willingness to buy product $i$.
\item[(ii)] A delivery company who wants to plan driver routes may not know the average traffic/congestion of road segments. Each segment $e$ (an edge in the graph) has a true average travel time $b_e$, but any given day it is $b_e$ + noise. One can formulate shortest paths as an LP where $\V{b}$ is the vector of edge lengths. Sampling $b_e$ corresponds to sending an observer to the road segment $e$ to observe how long it takes to traverse on a given day. 
\item[(iii)] A ride sharing company (e.g. Uber) wants to decide a set of prices for rides request but it may not know customers' likelihood of accepting the prices $c_i$. Sampling $c_i$ in this setting corresponds to posting different prices to collect information. 
\item[(iv)] For the purpose of recommending the best route in real time, a navigation App, e.g., Waze\footnote{www.waze.com}, may want to collect traffic information or route information from distributed driver via their App. %While the quality of such information varies, the information has been found very useful 
\end{itemize}

%\yc{I think we can safely remove the following discussion on crowdsourcing. They are not that relevant. Maybe we can keep the waze example somewhere in the intro?}[commented out.]
%
%Our model is motivated by the practice of actively acquiring information from people in crowdsourcing systems \citep{howe2006rise}, especially in crowdsensing \citep{reddy2010recruitment,ganti2011mobile,koutsopoulos2013optimal,chon2012automatically,gao2014jigsaw}, where noisy information can be elicited quickly and provide valuable information for many applications.
%Recently, crowdsensing has become a popular solution to companies or organization for collecting information from crowds conveniently.
%
%
%The idea of using crowdsourcing or crowdsensing has been used for information collection for different goals. For example, in \cite{abernethy2015low}, data are actively purchased from crowd workers for running online learning algorithms. In \cite{gao2013online}, the similar idea of querying data from crowdsourcing has been implemented within the context of online decision making. \cite{chon2012automatically} uses the idea of crowdsensing to characterize places based on data collected from smartphones. Another example is to use crowdsensing data to help with indoor plan reconstruction [\cite{gao2014jigsaw}].
%
%
%

\section{Change of Distribution Lemma}
Some of our lower bound proofs are based on the work \cite{chen17nearly}. For self-containedness, we restate some of the lemmas in \cite{chen17nearly}.

A key element to derive the lower bounds is the Change of Distribution lemma, which was first formulated in \cite{kaufmann2016complexity} to study best arm identification problem in multi-armed bandit model. The lemma provides a general relation between the expected number of draws and Kullback-Leibler divergences of the arms distributions. The core elements of the model that the lemma can be applied to are almost the same as the classical bandit model. We will state it here and explain the applicability of our setting. In the bandit model, there are $n$ arms, with each of them being characterized by an unknown distribution $\nu_i, i=1,2,...,n.$ The bandit model consists of a sequential strategy that selects an arm $a_t$ at time $t$. Upon selection, arm $a_t$ reveals a reward $z_t$ generated from its corresponding distribution. The rewards for each arm form an i.i.d. sequence. The selection strategy/algorithm invokes a stopping time $T$ when the algorithm will terminate and output a solution. To present the lemma, we define filtration $(\mathcal{F}_t)_{t\ge 0}$ with
 $\mathcal{F}_t = \sigma(a_1, z_1, \dots, a_t, z_t)$.
  
  If we consider a LP instance $\mathcal{I}$ with unknown parameters $d$ as a bandit model, an unknown parameter $d_i$ will correspond to an arm $a$ in the Lemma~\ref{ChangeDistr}, and thus $\nu_a$ is just the Gaussian distribution with mean $d_i$ and variance $1$ (both being unknown). Each step, we select one coefficient to sample with. Then we will be able to safely apply the Change of Distribution lemma to our setting. The lemma can be stated as follows in our setting of the problem.

%The goal of the strategy is to select the top K arms w.r.t. the average reward. A delta-correct algorithm in the bandit model is an algorithm that correctly identifies the top K arms with probability at least 1-delta.

\begin{lemma} \label{lemma:change_distr}
  Let $\mathcal{A}$ be a $(\delta,\epsilon_1,\epsilon_2)$-correct algorithm with $\delta \in (0, 0.1)$.
  Let $\mathcal{I},\mathcal{I}'$ be two LP instances that are equal on all known parameters, and let $d,d'$ be their respective vectors of unknown parameters.
  Suppose each instance has samples distributed Gaussian with variance $1$.
  Suppose $OPT(\mathcal{I}; \epsilon_1,\epsilon_2)$ and $OPT(\mathcal{I}'; \epsilon_1,\epsilon_2)$ are disjoint.
  Then letting $\tau_i$ be the number of samples $\mathcal{A}$ draws for parameter $d_i$ on input $\mathcal{I}$, we have
    \[ \mathbb{E} \sum_i \tau_i(d_i - d_i')^2 \geq 0.8 \ln\frac{1}{\delta} . \]
\end{lemma}
\begin{proof}
 % \bo{Sketch -- Sherry, please fill in missing details?}
  We use a result on bandit algorithms by \cite{kaufmann2016complexity}, which is restated as follows.
  %\bo{I think here we should replace the below with more exactly what that paper states as their lemma. Then, claim that for two normal distributions centered at $d_i$ and $d_i'$ with variance $1$, the KL-divergence is at most $\frac{1}{2}(d_i-d_i')^2$, either citing this fact from somewhere directly, or using Pinker's inequality and then citing a fact about total variation distance.}
  \begin{lemma}[\cite{kaufmann2016complexity}] \label{ChangeDistr}
  Let $\nu$ and $\nu'$ be two bandit models with $n$ arms such that for all arm $a$, the distribution $\nu_a$ and $\nu_a'$ are mutually absolutely continuous. For any almost-surely finite stopping time $T$ with respect to $(\mathcal{F}_t)$,
  \[
  \sum_{i=1}^n \mathbb{E}_{\nu}[N_a(T)] KL(\nu_a, \nu_a') \ge \sup_{\mathcal{E} \in \mathcal{F}_T} d(\Pr_\nu(\mathcal{E}), \Pr_{\nu'}(\mathcal{E})),
  \]
  where $d(x,y) = x \ln (x/y)+ (1-x)\ln\left((1-x)/(1-y)\right)$ is the binary relative entropy function, $N_a(T)$ is the number of samples drawn on arm $a$ before time $T$ and $KL(\nu_a, \nu_{a'})$ is the KL-divergence between distribution $\nu_a$ and $\nu_{a'}$.
  \end{lemma}
Let $\mathcal{I}$ and $\mathcal{I}'$ be the two bandit models in Lemma~\ref{ChangeDistr}. Applying above lemma we have
  \begin{align*}
  &\sum_{i=1}^n \mathbb{E}_{\mathcal{A}, \mathcal{I}}[\tau_i] KL(\mathcal{N}(d_i, 1), \mathcal{N}(d_i',1)) \\
  \ge & d(\Pr_{\mathcal{A}, \mathcal{I}}(\mathcal{E}), \Pr_{\mathcal{A}, \mathcal{I}'}(\mathcal{E})), \textrm{ for all } \mathcal{E} \in \mathcal{F}_T,
  \end{align*}
where $\mathcal{N}(\mu, \sigma)$ is the Gaussian distribution with mean $\mu$ and variance $\sigma$,  $\Pr_{\mathcal{A}, \mathcal{I}}[\mathcal{E}]$ is the probability of event $\mathcal{E}$ when algorithm $\mathcal{A}$ is given input $\mathcal{I}$, and $\mathbb{E}_{\mathcal{A}, \mathcal{I}}[X]$ is the expected value of random variable $X$ when algorithm $\mathcal{A}$ is given input $\mathcal{I}$.  According to the result in~\cite{duchi2007derivations}, the KL-divergence for two Gaussian distribution with mean $\mu_1,\mu_2$ and variance $\sigma_1, \sigma_2$ is equal to 
\[
\log \frac{\sigma_2}{\sigma_1} + \frac{ \sigma_1^2 + (\mu_1 - \mu_2)^2}{2 \sigma_2^2}.
\]
Thus we have $KL(\mathcal{N}(d_i, 1), \mathcal{N}(d_i',1)) =  \frac{1}{2} (d_i -d_i')^2$. We further define event $\mathcal{E}$ to be the event that algorithm $\mathcal{A}$ finally outputs a solution in set $OPT(\mathcal{I}; \varepsilon_1, \varepsilon_2)$, then since $\mathcal{A}$ is $(\delta, \varepsilon_1, \varepsilon_2)$-correct and $OPT(\mathcal{I}; \varepsilon_1, \varepsilon_2)$ is disjoint from $OPT(\mathcal{I}'; \varepsilon_1, \varepsilon_2)$, we have $\Pr_{\mathcal{A}, \mathcal{I}}(\mathcal{E}) \ge 1-\delta$ and $\Pr_{\mathcal{A}, \mathcal{I}'}(\mathcal{E})\le \delta$. Therefore 
\[
  \sum_{i=1}^n \mathbb{E}_{\mathcal{A}, \mathcal{I}}[\tau_i] \frac{1}{2} (d_i - d_i')^2 \ge  d(1-\delta, \delta)  \ge 0.4 \ln \delta^{-1}.
  \]
%  \bo{and use the proof you have below for $d(1-\delta,\delta) \geq 0.4\ln(1/\delta)$, and move the $2$ to the right hand side.}
  The last step uses the fact that for all $0< \delta <0.1$,
\[
d(1-\delta, \delta) = (1-2\delta) \ln \frac{1-\delta}{\delta} \ge 0.8 \ln \frac{1}{\sqrt{\delta}} = 0.4 \ln \delta^{-1}.
\]

\end{proof}

\section{The Unknown Constraints Case}

\subsection{Proof for Theorem~\ref{lb2:unknownb}} \label{app:unknownb_lower}
\unknownbLower*
Let $\mathcal{A}$ be a $(\delta, \varepsilon_1, \varepsilon_2)$-correct algorithm. For a positive integer $n$, consider the following LP instance $\mathcal{I}$ with $n$ variables and $n$ constraints,
\begin{align*}
\max & \quad x_1\\
\textrm{s.t. } & \quad   x_1 \le C,\\
		& \quad   x_1 + x_i \le C, \quad \forall  2\le i \le n, \\
		& \quad \Sol \ge \V{0}.\\
\end{align*}
Clearly the optimal solution is $x^*_1 = C$ and $x^*_i = 0$ for $i>1$. Every constraint is a binding constraint.  Now we prove that for any $k\in [n]$, algorithm $\mathcal{A}$ should take at least $\Omega\left(\ln(1/\delta)\cdot \max\{\varepsilon_1,\varepsilon_2\}^{-2}\right)$ for the $k^{th}$ constraint. We construct a new LP $\mathcal{I}'$ by subtracting the right-hand side of $k^{th}$ constraint by $2(\varepsilon_1 + \varepsilon_2)$. Then $OPT(\mathcal{I}; \varepsilon_1, \varepsilon_2)$ and $OPT(\mathcal{I}'; \varepsilon_1, \varepsilon_2)$ must be disjoint, since for any $\Sol \in OPT(\mathcal{I}'; \varepsilon_1, \varepsilon_2)$, $\Sol$ will not violate the $k^{th}$ constraint of $\mathcal{I}'$ by more than $\varepsilon_2$, $$x_1 \le C - 2(\varepsilon_1 + \varepsilon_2) + \varepsilon_2 < C - 2\varepsilon_1,$$ which means that $\Sol \notin OPT(\mathcal{I}; \varepsilon_1, \varepsilon_2)$. According to Lemma~\ref{lemma:change_distr}, 
\[
\E[\tau_k] \cdot 4(\varepsilon_1 + \varepsilon_2)^2 \ge 0.8 \ln(1/\delta)
\]
And since $2 \max\{\varepsilon_1, \varepsilon_2\} \ge \varepsilon_1 + \varepsilon_2$, 
\[
\E[\tau_k] = \Omega( \max\{\varepsilon_1, \varepsilon_2\}^{-2} \cdot \ln(1/\delta)).
\]

\subsection{Proof for Theorem~\ref{main}} \label{app:Ellipsoid-UCB_proof}

        Recall that our algorithm and the sample complexity theorem works as follows:
        \begin{algorithm}[!h]                      % enter the algorithm environment
\caption{Modified ellipsoid algorithm}
\begin{algorithmic}
    \STATE Let $\mathcal{E}^{(0)}$ be the initial ellipsoid containing the feasible region.\\
    \STATE Draw one sample for each $b_i$, $i\in [m]$.\\
    \STATE Let $k = 0$ and $t=m$.\\
    \STATE Let $T_i(t) = 1$ for all $i$. \\%\bo{added this line}\\
    \WHILE{stopping criterion is not met\footnotemark}
            \STATE Let $\V{x}^{(k)}$ be the center of $\mathcal{E}^{(k)}$
            \STATE Call UCB method to get constraint $i$ or ``feasible''
            \IF{$\V{x}^{(k)}$ is feasible}
%                \STATE Let $\V{x} \gets \V{x}^{(k)}$ if $\V{x}^{(k)}$ is the first feasible point, or $\V{x}^{(k)}$ is better, i.e. $\V{c}^T \V{x}^{(k)} > \V{c}^T \V{x}$. \yc{We need to initialize $\V{x}$ earlier on, I think.} \sherry{change the condition of update}
                \STATE Let $\V{x} \gets \V{x}^{(k)}$ if $\V{x}$ is not initialized or $\V{c}^T \V{x}^{(k)} > \V{c}^T \V{x}$. %\bo{Updated this line.}
%                \STATE Compare $\Sol$ with $\V{x}^{(k)}$ \bo{$\Sol \gets \V{x}^{(k)}$ Correct?} \sherry{I think we cannot replace $\Sol$ with $\V{x}^{(k)}$ directly cause ellipsoid does not completely remove points $c^T x \le c^T x^{(k)}$}
                \STATE $\V{y} \gets -\V{c}$
            \ELSE
                \STATE $\V{y}\gets \V{A}_{i}^T$
            \ENDIF
            \STATE Let $\mathcal{E}^{(k+1)}$ be the minimal ellipsoid that contains $\mathcal{E}^{(k)} \cap \{\V{t} ~:~ \V{y}^T \V{t} \le \V{y}^T\V{x}^{(k)}\}$\\
            \STATE Let $k \gets k+1$
    \ENDWHILE
    \STATE Output $\Sol$ or ``failure'' if it was never set.
\end{algorithmic}
\end{algorithm}
\footnotetext{Our stopping criterion is exactly the same as in the standard ellipsoid algorithm, for which there are a variety of possible criteria that work. In particular, one is $\sqrt{\V{c}^T \V{P}^{-1} \V{c}} \leq \min\{\varepsilon_1,\varepsilon_2\}$, where $P$ is the matrix corresponding to ellipsoid $\mathcal{E}^{(k)}$ as discussed above.}

\begin{algorithm}[!h]                      % enter the algorithm environment
\caption{UCB-method}
\begin{algorithmic}
    \STATE Input $\V{x}^{(k)}$
    \STATE Set $\delta' = \left(\frac{\delta}{20m}\right)^{2/3}$
    \LOOP
            \STATE  1. Let $j$ be the constraint with the largest index,
            \[
                j = \arg\max_i \V{A}_i \V{x}^{(k)} - \widehat{b}_{j,T_j(t)} + U_i(T_i(t)),
            \]
            where $U_i(s) = 3\sqrt{\frac{2\sigma_i^2\log\left(\log(3s/2)/\delta'\right)}{s}}$ and $\widehat{b}_{j,T_j(t)}$ as in Definition~\ref{def:b-hat}.
            \STATE 2. If $\V{A}_j \V{x}^{(k)} - \widehat{b}_{j,T_j(t)} - U_j(T_j(t))> 0$ return $j$.
            \STATE 3. If $\V{A}_j \V{x}^{(k)} - \widehat{b}_{j,T_j(t)} + U_j(T_j(t))< 0$ return ``feasible''.
            \STATE 4. If $U_j(T_i(t))< \varepsilon_2/2$ return ``feasible''.
            \STATE 5. Let $t \gets t+1$
            \STATE 6. Draw a sample of $b_j$.
            \STATE 7. Let $T_j(t) = T_j(t-1) + 1$.
            \STATE 8. Let $T_i(t) = T_i(t-1)$ for all $i\neq j$.
    \ENDLOOP
\end{algorithmic}
\end{algorithm}

      \sampleEllipsoidUCB*  
        
        Our analysis is inspired by the techniques used in \cite{jamieson2014lil}.
        The following lemma is the same as Lemma 3 in \cite{jamieson2014lil}, and is simplified by setting $\varepsilon=1/2$. We choose 1/2 only for simplicity. It will not change our result asymptotically. The constant in this lemma can be optimized by selecting parameters carefully.
        
        \begin{lemma} \label{UniformBound}
        Let $X_1, X_2, \dots$ be i.i.d. sub-Gaussian random variables with scale parameter $\sigma$ and mean $\mu$. With probability at least $1-20\cdot\delta^{3/2}$, we have for all $t\ge 1$, $\left|\frac{1}{t}\sum_{s=1}^t X_s - \mu \right| \le L(t,\delta)$, where $L(t,\delta)=3\sqrt{\frac{2\sigma^2\log\left(\log(3t/2)/\delta\right)}{t}}$.
        \end{lemma}
        
        Define event $\mathcal{A}$ to be the event that $\left|\B_{i,t} - b_i \right| \le U_i(t)$ for all $t\ge 0$ and $i\in[m]$. Since our definition of $U_i(t)$ is the same as $L(t, (\delta/20m)^{2/3})$ in Lemma~\ref{UniformBound} with scale parameter $\sigma_i$, the probability that event $\mathcal{A}$ holds is at least $1-\delta$ according to union bound.
        
         We prove the correctness and the sample number of the algorithm conditioning on that $\mathcal{A}$ holds.
        \paragraph{Correctness:} We first prove that the output of our algorithm satisfies relaxed feasibility and relaxed optimality when $\mathcal{A}$ holds. If our UCB-method always gives correct answer, the ellipsoid algorithm will be able to find an $\varepsilon_1$-suboptimal solution. So we only need to prove the correctness of the UCB-method.
        \begin{itemize}
        \item When UCB method returns a violated constraint $j$ in line 2, it is indeed a violated one: since $|\B_{j,T_j(t)}-b_j|\le U_j(T_j(t))$,
        \begin{align*}
        &\V{A}_j \V{x}_k - b_j \\
        \ge & \V{A}_j \V{x}_k - \B_{j,T_j(t)} - U_j(T_j(t))\\
        > & 0.
        \end{align*}
        \item When it returns ``feasible'' in line 3, no constraint is violated:
        \begin{align*}
        &\V{A}_i \V{x}_k - b_i \\
        \le &\V{A}_i \V{x}_k - \B_{i,T_i(t)} + U_i(T_i(t))\\
        \le &\V{A}_j \V{x}_k - \B_{j,T_j(t)} + U_j(T_j(t))\\
        <& 0, \quad \forall i\in[m].
        \end{align*}
        \item When it returns ''feasible'' in line 4, no constraint is violated more than $\varepsilon_2$:
        \begin{align*}
        &\V{A}_i \V{x}_k - b_i \\
        \le & \V{A}_i \V{x}_k - \B_{i,T_i(t)} + U_i(T_i(t))\\
        \le & \V{A}_j \V{x}_k - \B_{j,T_j(t)} + U_j(T_j(t))\\
        \le & \V{A}_j \V{x}_k - \B_{j,T_j(t)} - U_j(T_j(t)) + 2 U_j(T_j(t))\\
        \le & 0 + \varepsilon_2, \quad \forall i \in [m].
        \end{align*}
        \end{itemize}
        Therefore the relaxed feasibility should be satisfied and the relaxed optimality is guaranteed by ellipsoid algorithm.
        
        \paragraph{Number of samples:} We bound the number of samples used on each constraint separately. The number of samples used on constraint $i$ can be stated as
        the maximum $T_i(t)$ where $t$ is a mini-stage in which a sample of $b_i$ is drawn. We bound $T_i(t)$ by showing that $U_i(T_i(t))$ should be larger than a certain value if $b_i$ is sampled at mini-stage $t$. This immediately give us an upper bound of $T_i(t)$ since $U_i(t)$ is a decreasing function of $t$.
        Suppose $b_i$ is sampled at mini-stage $t$ in ellipsoid iteration $k$. Let $i^*$ be the constraint with largest violation. Conditioning on $\mathcal{A}$ holds, the fact that constraint $i$ have a larger index than $i^*$ gives
        \begin{align}\label{eq1}
        & V_i(k) + 2 U_i(T_i(t)) \notag\\
        \ge & \V{A}_i \V{x}_k - \B_{i,T_i(t)} + U_i(T_i(t)) \notag\\
        \ge & \V{A}_{i^*} \V{x}_k - \B_{i^*,T_{i^*}(t)} + U_{i^*}(T_{i^*}(t)) \notag\\
        \ge & V_{i^*}(k).
        \end{align}
        which implies $2U_i(T_i(t)) \ge V^*(k)-V_i(k)$.
        Now look at line 2 in UCB-method. If a sample of $b_i$ is drawn, we should not quit in this step. So if $V_i(k)>0$, we must have
        \begin{align}\label{eq2}
        &V_i(k) - 2U_i(T_i(t))\notag\\
        \le &\V{A}_i \V{x}_k - \B_{i,T_i(t)} - U_i(T_i(t))\notag\\
        \le & 0.
        \end{align}
        Similarly, because of line 3 in UCB-method, if $V_i(k)\le 0$, it should be satisfied that
        \begin{align}\label{eq3}
        &V_i(k) + 2 U_i(T_i(t))\notag\\
        \ge & \V{A}_i \V{x}_k - \B_{i,T_i(t)} + U_i(T_i(t)) \notag\\
        \ge & 0.
        \end{align}
         Putting inequality~\eqref{eq1}, \eqref{eq2} and \eqref{eq3} and $U_i(T_i(t))\ge \varepsilon_2/2$ together, we get the conclusion that $2U_i(T_i(t))\ge \max\{V^*(k)-V_i(k), |V_i(k)|, \varepsilon_2/2\} = gap_{i,\varepsilon_2/2}(k)$ should be satisfied if we draw a sample of $b_i$ at mini-stage $t$ in ellipsoid iteration $k$.
        
        Then we do some calculation,
        \begin{align}
        &2U_i(T_i(t)) \ge gap_{i,\varepsilon_2/2}(k) \notag\\
        \Rightarrow \quad & 6\sqrt{\frac{2\sigma_i^2 \log(\log(3T_i(t)/2)/\delta'}{T_i(t)}} \ge gap_{i,\varepsilon_2/2}(k) \notag\\
        \Rightarrow \quad & \frac{\log(\log(3T_i(t)/2)/\delta')}{T_i(t)} \ge \frac{gap_{i,\varepsilon_2/2}^2(k)}{72\sigma_i^2} \notag\\
        \Rightarrow \quad & T_i(t) \le \frac{108\sigma_i^2}{gap^2_{i,\varepsilon_2/2}(k)} \log\left(\frac{20m}{\delta}\right) \notag\\
        &+ \frac{72\sigma_i^2}{gap^2_{i,\varepsilon_2/2}(k)}\log\log\left(\frac{108\sigma_i^2}{gap^2_{i,\varepsilon_2/2}(k) \delta'}\right). \label{result1}
        \end{align}
        In the last step, we use the fact that for $0 < \delta \le 1$, $c>0$,
        \begin{align*}
        &\frac{1}{t}\cdot\log\left(\frac{\log(3t/2)}{\delta}\right)\ge c \\
        \Rightarrow \quad & t\le \frac{1}{c} \log\left(\frac{2\log(3/(2c\delta))}{\delta}\right).
        \end{align*}
        
        Take maximum of \eqref{result1} over all $k$ and according to the definition of $\Delta_{i, \varepsilon_2/2}$,
        \begin{align*}
        T_i(t) \le \frac{108\sigma_i^2}{\Delta_{i,\varepsilon_2/2}^2} \log\left(\frac{20m}{\delta}\right) + \\ \frac{72\sigma_i^2}{\Delta_{i,\varepsilon_2/2}^2}\log\log\left(\frac{108\sigma_i^2}{\Delta_{i,\varepsilon_2/2}^2 \delta'}\right)
        \end{align*}
        Therefore the overall number of samples is at most
        \[
        O\left(\sum_i \frac{\sigma_i^2}{\Delta_{i,\varepsilon_2/2}^2}\log\frac{m}{\delta} + \sum_i \frac{\sigma_i^2}{\Delta_{i,\varepsilon_2/2}^2}\log\log\left(\frac{\sigma_i^2}{\Delta_{i,\varepsilon_2/2}^2}\right)\right).
        \]

\section{Proofs for  the Unknown Objective Function Case}

\subsection{Proof for Theorem~\ref{thm:instance_lowerbound}}

        We restate the instance-wise lower bound for unknown objective function LP problems.
        
        \unknowncInstLowerbound*

        Let $\mathcal{I}$ be a LP instance $\max_{\{\V{x}:\V{A}\V{x}\le \V{b}\}} \V{c}^T \V{x}$, and $\mathcal{A}$ be a $\delta$-correct algorithm, where $0<\delta<0.1$. Define $t_i$ to be the expected number of samples that algorithm will draw for $c_i$ when the input is $\mathcal{I}$. 
        
        We only need to show that $5\V{t}/\ln(1/ \delta)$ is a feasible solution of the convex program~\eqref{low} that computes $Low(\mathcal{I})$. 
        %Because it  immediately follows that
        %\[
        %\sum_{i=1}^n t_i = \frac{d(1-\delta, \delta)}{2}\sum_{i=1}^n 2 t_i /d(1-\delta, \delta) \ge \frac{d(1-\delta, \delta)}{2} Low(\mathcal{I}) = \Omega(Low(\mathcal{I}) \ln \delta^{-1}).
        %\] 
        %The last step uses the fact that for all $0< \delta <0.1$,
        %\[
        %d(1-\delta, \delta) = (1-2\delta) \ln \frac{1-\delta}{\delta} \ge 0.8 \ln \frac{1}{\sqrt{\delta}} = 0.4 \ln \delta^{-1}.
        %\]
        
        Consider a constraint in~\eqref{low}
        \[
        \sum_{i=1}^n \frac{(s_i^{(k)} -x_i^*)^2}{\tau_i} \le \left( \V{c}^T (\V{x}^* - \V{s}^{(k)}) \right)^2,
        \]
        where $\Sol^*$ is the optimal solution of $\mathcal{I}$ and $\V{s}^{(k)}$ is a corner point of the feasible region of $\mathcal{I}$. To prove that $5\V{t}/\ln(1/ \delta)$ satisfies this constraint, we will construct a new LP instance $\mathcal{I}_{\V{\Delta}}$ by adding $\V{\Delta}$ to the objective function $\V{c}$, such that $\V{s}^{(k)}$ becomes a better solution than $\Sol^*$. We construct vector $\V{\Delta}$ as follows,
        \[
        \Delta_i = \frac{D(x^*_i - s^{(k)}_i)}{t_i}, \quad \textrm{ and } \quad D=\frac{-2 \V{c}^T (\Sol^* - \V{s}^{(k)})}{\sum_{i=1}^n \frac{(s^{(k)}_i - x^*_i)^2}{t_i}}.
        \]
        It is not difficult to verify that $\Sol^*$ is no longer the optimal solution of $\mathcal{I}_{\V{\Delta}}$:
        \begin{align*}
        &\langle \V{c}+\V{\Delta}, \Sol^* - \V{s}^{(k)} \rangle \\
        = & \langle \V{c}, \Sol^* - \V{s}^{(k)} \rangle
        + \langle \V{\Delta}, \Sol^* - \V{s}^{(k)} \rangle \\
        = & \langle \V{c}, \Sol^* - \V{s}^{(k)} \rangle - \\
        & \sum_{i=1}^n  \frac{2 \V{c}^T (\Sol^* - \V{s}^{(k)})}{\sum_{i=1}^n \frac{(s^{(k)}_i - x^*_i)^2}{t_i}} \cdot \frac{x^*_i - s^{(k)}_i }{t_i} \cdot (x^*_i - s^{(k)}_i )\\
        = & -\langle \V{c}, \Sol^* - \V{s}^{(k)} \rangle \\
        < &0.
        \end{align*}
        %And because $\mathcal{A}$ is a $\delta$-correct algorithm, $\Pr[\mathcal{A}(\mathcal{I}) = \Sol^*] \ge 1-\delta$ and $\Pr[\mathcal{A}(\mathcal{I}_{\V{\Delta}}) = \Sol^*] \le \delta$. So if we define event $\mathcal{E}$ to be the event that algorithm $\mathcal{A}$ outputs $\Sol^*$, and apply Lemma~\ref{ChangeDistr} on $\mathcal{I}$ and $\mathcal{I}_{\V{\Delta}}$, we get
        %\[
        %\sum_{i=1}^n t_i \cdot \frac{1}{2}\Delta_i^2 \ge d\left(\Pr[\mathcal{A}(\mathcal{I}) = \Sol^*], \Pr[\mathcal{A}( \mathcal{I}_{\V{\Delta}}) = \Sol^*]\right) \ge d(1-\delta, \delta).
        %\]
         %Here the second inequality is derived from the property of binary relative entropy function. Plugging in 
         %the value of $\Delta_i$, 
         Then by Lemma~\ref{lemma:change_distr},
         \begin{align*}
         0.8 \ln (1/\delta) &\le  \sum_{i=1}^n t_i \cdot \Delta_i^2\\
          &= \sum_{i=1}^n t_i \cdot \left(\frac{D(x^*_i - s^{(k)}_i)}{t_i}\right)^2 \\
         &=  \sum_{i=1}^n \frac{(x^*_i - s^{(k)})^2}{t_i} \cdot D^2 \\
         &=  \sum_{i=1}^n \frac{(x^*_i - s^{(k)})^2}{t_i}  \cdot \left(\frac{-2 \V{c}^T (\Sol^* - \V{s}^{(k)})}{\sum_{i=1}^n \frac{(s^{(k)}_i - x^*_i)^2}{t_i}}\right)^2\\
         &= 4 \cdot \frac{( \V{c}^T (\Sol^* - \V{s}^{(k)}))^2}{\sum_{i=1}^n \frac{(s^{(k)}_i - x^*_i)^2}{t_i}},
         \end{align*}
         which is equivalent to
         \[
         \sum_{i=1}^n \frac{(s^{(k)}_i - x^*_i)^2}{5 t_i/\ln(1/\delta)} \le ( \V{c}^T (\Sol^* - \V{s}^{(k)}))^2.
         \]
         Therefore $5\V{t}/\ln(1/\delta)$ is a feasible solution of the convex program~\eqref{low}, which completes our proof.

\subsection{Proof for Theorem~\ref{thm:worst_lowerbound}} \label{app:unknownc_worst}

We prove the worst case lower bound for unknown $\V{c}$ case.

       \unknowncWorstLowerbound*
       
        The following lemma will be used in the construction of desired LP instances.
        \begin{lemma}\label{subsets}
        Let $n$ be a positive integer. There exists a constant $c$, a positive integer $l=\Omega(n)$ and $z=2^{cn}$  sets $W_1, \dots, W_z \subseteq [n]$ such that 
        \begin{itemize}
        \item For all $i \in [z]$, we have $|W_i| = l = \Omega(n)$.
        \item For all $i\neq j$, $|W_i \cap W_j| \le l/2$.
        \end{itemize}
        \end{lemma}
        \begin{proof}
        Define  $l = n/10$. Let each $W_i$ be  a uniformly random subset of $[n]$ with size $l$. Then it is satisfied that
        \[
          \Pr[\vert W_i \cap W_j \vert > l/2]\le 2^{-\Omega(n)}
          \]
        for all $1\le i, j \le n, i\neq j$. 
        So we can choose sufficiently small $c$ such that
          \[
          \Pr[\exists i\neq j, \vert W_i \cap W_j \vert > l/2]\le z^2 2^{-\Omega(n)}< 1,
          \]
        which implies the existence of a desired sequence of subsets.
          \end{proof}
          
          Now for any $\delta$-correct algorithm $\mathcal{A}$, we prove the existence of  LP instances $\mathcal{I}_1, \mathcal{I}_2, \dots$, which all have $n$ variables. 
          
          For simplicity, all the linear program instances we construct in this proof share the same feasible region, which we define as follows. Let $W_1,\dots,W_z \subseteq [n]$ be the sequence of subsets in Lemma~\ref{subsets}. For a subset $W\subseteq [n]$,  we define a point $\V{p}^W$ 
          \begin{align*}
          p^W_i = \left\{\begin{array}{ll}
          			1, & \textrm{ if } i \in W;\\
        			0, & \textrm{ otherwise. }
        			\end{array}
        			\right.
        \end{align*}
          The feasible region we are going to use throughout this proof is the convex hull of $\V{p}^{W_1}, \dots, \V{p}^{W_z}$. 
         
         To find a desired LP instance $\mathcal{I}_k$, we first choose an arbitrary constant $\Delta_k$.  We construct $z$ different LP instances $\mathcal{I}_{\Delta_k, W_1}, \dots, \mathcal{I}_{\Delta_k, W_z}$ and show that at least one of them satisfies the condition in the theorem. Define the objective function $\V{c}^{W_j}$ of $\mathcal{I}_{\Delta_k, W_j}$ to be
          \begin{align*}
          c^{W_j}_i = \left\{\begin{array}{ll}
          			\Delta_k, & \textrm{ if } i \in W_j;\\
        			-\Delta_k, & \textrm{ otherwise. }
        			\end{array}
        			\right.
        \end{align*}
         Then clearly the optimal solution of $\mathcal{I}_{\Delta_k, W_j}$ is point $\V{p}^{W_j}$. We define 
        $\Pr[\mathcal{A}(\mathcal{I}_{\Delta, W_i}) = \V{p}^{W_j}]$ to be the probability that algorihtm $\mathcal{A}$ outputs $\V{p}^{W_j}$ when the input is $\mathcal{I}_{\Delta, W_i}$. Then we have 
         \[
         \Pr[\mathcal{A}(\mathcal{I}_{\Delta, W_i}) = \V{p}^{W_i}] \ge 1-\delta,
         \]
         and
         \[
         \sum_{j: j\neq i} \Pr[\mathcal{A}(\mathcal{I}_{\Delta, W_i}) = \V{p}^{W_j}] \le \delta.
         \]
         Thus there must exists $W_k$ such that
         \[
        	\Pr[\mathcal{A}(\mathcal{I}_{\Delta, W_i}) = \V{p}^{W_k}] \le 2\delta/z.
         \]
        Let $T$ be the number of samples used by algorithm $\mathcal{A}$ when the input is $\mathcal{I}_{\Delta_k, W_k}$. Since $\mathcal{A}$ is a $\delta$-correct algorithm, $\Pr[\mathcal{A}(\mathcal{I}_{\Delta, W_k}) = \V{p}^{W_k}] \ge 1 - \delta> 0.9$.  So if we define event $\mathcal{E}$ to be the event that $\mathcal{A}$ outputs $\V{p}^{W_k}$ and apply Lemma~\ref{ChangeDistr},
         \begin{align*}
         &\mathbb{E}[T]\cdot (2 \Delta^2) \\
         \ge & d\left(\Pr[\mathcal{A}(\mathcal{I}_{\Delta, W_k}) \V{p}^{W_k}],\Pr[\mathcal{A}(\mathcal{I}_{\Delta, W_i}) = \V{p}^{W_k}]\right) \\
        \ge & \Omega(\ln(z/\delta)\\
        = & \Omega(\ln z + \ln(1/\delta)).
        \end{align*}
        Here we use the following property of $d(1-\delta,\delta)$ function: for $0<\delta<0.1$, $d(1-\delta,\delta)\ge 0.4 \ln(1/\delta)$. So we get a lower bound for $\mathbb{E}[T]$, $$\mathbb{E}[T] \ge \Omega\left(\Delta^{-2}(\ln z + \ln(1/\delta))\right).$$
        
        Meanwhile if we look at the Instance Lower Bound, $Low(\mathcal{I}_{\Delta, W_k})$, 
        	\begin{align*}
                 \min_\tau \quad & \sum_{i=1}^n \tau_i \\
                s.t.\quad & \sum_{i=1}^n \frac{(p_i^{W_j} -p_i^{W_k})^2}{\tau_i} \le \langle \V{c}^{W_k}, (\V{p}^{W_k} - \V{p}^{W_j})\rangle^2, \forall j \\
                 &\tau_i\ge 0, 
                \end{align*}
        It is easy to verify that $\tau_i = \frac{8}{l\Delta^2}$ for all $i$ is a feasible solution. So we have $Low(\mathcal{I}_{\Delta, W_k})=\Theta(\frac{8n}{l\Delta^2}) = \Theta(\Delta^{-2})$. Therefore the number of samples that $\mathcal{A}$ will use on $\mathcal{I}_{\Delta, W_k}$ is $\Omega\left(Low(\mathcal{I}_{\Delta, W_k})(\ln z+\ln (\delta^{-1})) \right)$ in expectation.
        
        By simply setting $\Delta_k = \frac{1}{k}$, we will get an infinite sequence of LP instances as stated in the theorem.  

\subsection{Proof for Theorem~\ref{thm:unknownc}}

In this section, we prove the sample complexity of our successive elimination algorithm for unknown $c$ case.

\unknowncSample*

The following lemma will be used in our proof.
        \begin{lemma}\label{gaussianLem}
            Given a set of Gaussian arms with unit variance and mean $c_1, \dots, c_n$. Suppose we take $\tau_i$ samples for arm $i$. Let $X_i$ be the empirical mean. Then for an arbitrary vector $\V{p}$, 
            \[
            \Pr\left[ |\V{p}^T \V{X} - \V{p}^T \V{c}|\ge \varepsilon\right] \le 2 \exp\left(-\frac{\varepsilon^2}{2 \sum p_i^2/\tau_i}\right)
            \]
            \end{lemma}
            \begin{proof}
            By definition, $\V{p}^T \V{X} - \V{p}^T \V{c}$ follows Gaussian distribution with mean $0$ and variance $\sum_i p_i^2/\tau_i$. 
            \end{proof}

        We define a good event $\mathcal{E}$ to be the event that $|(\V{x}-\V{y})^T (\widehat{\V{c}}^{(r)} - \V{c}))|\le \varepsilon^{(r)} /\lambda$ for all $r$ and $\V{x},\V{y} \in S^{(r)}$. According to Lemma~\ref{gaussianLem}, 
        \begin{align*}
        \Pr[\mathcal{E}] \ge 1 - \sum_r \sum_{\V{x}\in S^{(r)}} \sum_{\V{y} \in S^{(r)}}  2 \exp{- \frac{(\varepsilon/\lambda)^2}{2 \sum (x_i - y_i)^2/\tau_i}}.
        \end{align*}
        Since $\V{\tau}$ satisfies the constraints in~\eqref{lowall},
        \begin{align*}
        &\sum_r \sum_{\V{x}\in S^{(r)}} \sum_{\V{y} \in S^{(r)}}  2 \exp{- \frac{(\varepsilon/\lambda)^2}{2 \sum (x_i - y_i)^2/\tau_i}} \\
        \le & \sum_r \sum_{\V{x}\in S^{(r)}} \sum_{\V{y} \in S^{(r)}}  2 \exp\left(- \ln(2/\delta^{(r)})\right)\\
        = & \sum_r \sum_{\V{x}\in S^{(r)}} \sum_{\V{y} \in S^{(r)}} \delta^{(r)}\\
        \le & \delta
        \end{align*}
        Therefore $\Pr[\mathcal{E}] \ge 1-\delta$. 
        
        We first prove the correctness of the algorithm conditioning on $\mathcal{E}$.
        
        \begin{lemma} \label{correctness}
        When the good event $\mathcal{E}$ holds, the optimal LP solution $\V{x}^* = \max_{\V{A}\V{x} \le \V{b}} \V{c}^T \V{x}$ will not be deleted.
        \end{lemma}
        \begin{proof}
        Suppose to the contrary $\Sol^*$ is deleted in iteration $r$, i.e., $\Sol^* \in S^{(r)}$ but $\Sol^* \notin S^{(r+1)}$. Then according to~\eqref{eliminateCond}, when the objective function is $\widehat{\V{c}}^{(r)}$, $\Sol^*$ is at least $\varepsilon^{(r)}/2-2\varepsilon^{(r)}/\lambda$ worse than $\Sol^{(r)}$, 
        \[
        \langle \Sol^{(r)} - \Sol^* , \widehat{\V{c}}^{(r)}\rangle > \varepsilon^{(r)}/2+2\varepsilon^{(r)}/\lambda.
        \]
        By the definition of the optimal solution $\Sol^*$,
        \[
        \langle \V{c}, \Sol^* - \Sol^{(r)} \rangle > 0. 
        \]
        Combining the two inequalities will give
        \[
        \langle \V{c} - \widehat{\V{c}}^{(r)}, \Sol^* - \Sol^{(r)} \rangle > \varepsilon^{(r)}/2+2\varepsilon^{(r)}/\lambda >
        \varepsilon^{(r)}/\lambda,
        \]
        contradictory to that event $\mathcal{E}$ holds.
        \end{proof}
        
        We then bound the number of samples conditioning on $\mathcal{E}$.
        We first prove the following lemma.
        \begin{lemma} \label{correct2}
        When event $\mathcal{E}$ holds,  all points $\V{s}$ in set $S^{(r+1)}$ satisfies
        \[
        \langle \V{c}, \Sol^* - \V{s} \rangle < \varepsilon^{(r)}.
        \]
        after the $r^{th}$ iteration.
        \end{lemma}
        \begin{proof}
        Suppose when entering the $r^{th}$ iteration, there exists $\V{s}\in S^{(r)}$ such that 
        $\langle \V{c}, \Sol^* - \V{s} \rangle > \varepsilon^{(r)}$. Then since $\mathcal{E}$ holds and $\lambda = 10$,
        \begin{align*}
        \langle \V{c}, \Sol^* - \V{s} \rangle &> \langle \widehat{\V{c}}^{(r)}, \Sol^* - \V{s} \rangle - \varepsilon^{(r)}/\lambda \\
        & > (1 - 1/\lambda) \varepsilon^{(r)} \\
        & > \varepsilon^{(r)} /2 + 2 \varepsilon^{(r)} / \lambda.
        \end{align*} 
        By Lemma~\ref{correctness}, we have $\Sol^* \in S^{(r)}$.
        Therefore $\V{s}$ will be deleted in this iteration.
        \end{proof}
        
        Now consider a fixed iteration $r$.
        Let $\V{\tau}^*$ be the optimal solution of the convex program~\eqref{low} that computes $low(\mathcal{I})$. Define $\alpha = 32 \lambda^2 \ln(2/\delta^{(r)})$. We show that $\V{t} = \alpha \V{\tau}^*$ is a feasible solution in the convex program~\eqref{lowall} that computes $LowAll(S^{(r)}, \varepsilon^{(r)}, \delta^{(r)})$. For any $\V{x},\V{y}\in S^{(r)}$, 
        \begin{align*}
        \sum \frac{(x_i-y_i)^2}{t_i} &= \frac{1}{\alpha} \sum \frac{(x_i-y_i)^2}{\tau_i^*}\\
        & = \frac{1}{\alpha} \sum \frac{(x_i-x^*_i + x^*_i - y_i)^2}{\tau_i^*}\\
        & \le \frac{1}{\alpha} \sum \frac{2(x_i-x^*_i)^2 + 2(x^*_i - y_i)^2}{\tau_i^*}
        \end{align*}
        due to the fact that $(a+b)^2 \le 2a^2 + 2b^2$ for all $a,b \in \mathbb{R}$. 
        
        Since $\tau^*$ satisfies the constraints in $Low(\mathcal{I})$ function~\eqref{low},
        \begin{align*}
        & \frac{1}{\alpha} \sum \frac{2(x_i-x^*_i)^2 + 2(x^*_i - y_i)^2}{\tau_i^*} \\
         \le & \frac{2}{\alpha} \left((\V{c}^T(\Sol^*-\Sol))^2 + (\V{c}^T(\Sol^* - \V{y}))^2\right) 
        \end{align*}
        And because of Lemma~\ref{correct2},
        \begin{align*}
        & \frac{2}{\alpha} \left((\V{c}^T(\Sol^*-\Sol))^2 + (\V{c}^T(\Sol^* - \V{y}))^2\right) \\
        \le & \frac{4}{\alpha} (\varepsilon^{(r-1)})^2 = \frac{(\varepsilon^{(r)})^2}{2\lambda^2 \ln(2/\delta^{(r)})}.
        \end{align*}
        
        So we have proved that $\V{t} = \alpha \V{\tau}^*$ is a feasible solution of the convex program that computes $LowAll(S^{(r)}, \varepsilon^{(r)}, \delta^{(r)})$. Thus the number of samples used in iteration $r$, $\sum_{i=1}^n t^{(r)}_i$, is no more than 
        \begin{align*}
         &\sum_{i=1}^n t^{(r)}_i  \le  \sum_{i=1}^n t_i =  \alpha \sum_i \tau_i^* \\
        = &O(Low(\mathcal{I})(\ln |S^{(r)}| + \ln \delta^{-1} + \ln r ) 
        \end{align*}
        
        Conditioning on $\mathcal{E}$, the algorithm will terminate before $\lfloor \log(\Delta^{-1})\rfloor + 1$ iterations according to Lemma~\ref{correct2}. Therefore the total number of samples is
        \[
        O\left( Low(\mathcal{I}) \ln \Delta^{-1} ( \ln |S^{(1)}| + \ln \delta^{-1} + \ln \ln \Delta^{-1})\right).
        \]

\end{document}